\documentclass[]{eptcs}
\providecommand{\event}{GandALF 2022} 
\pdfoutput=1

\usepackage{amssymb,amsfonts}
\usepackage{amsthm}
\usepackage{graphicx} 
\usepackage{amsmath}
\usepackage{times}
\usepackage{enumerate}
\usepackage{todonotes}
\usepackage{multirow}
\usepackage{hyperref}
\usepackage{bm}
\usepackage{relsize}
\usepackage{tablefootnote}

\usepackage[font=small]{caption}
\DeclareMathAlphabet{\mathcal}{OMS}{cmsy}{m}{n} 

\usepackage{hyperref} 

\usepackage{multirow}

\usepackage{enumerate}
\usepackage{hyperref}
\hypersetup{%
  colorlinks=true,           
  allcolors=blue!70!black,   
  pdfstartview=Fit,          
  breaklinks=true,
  }

\newcommand{\quant}{\mathbb{Q}}
\newcommand{\stam}[1]{}

\newcommand{\range}{\mathsf{range}}

\newcommand{\hadar}[1]{\colorbox{green!30!white}{......}\footnote{\colorbox{green!30!white}{\textsc{hf:}}~#1\colorbox{green!30!white}{.}}}

\newcommand{\sarai}[1]{\colorbox{pink!70!white}{......}\footnote{\colorbox{pink!70!white}{\textsc{sr:}}~#1\colorbox{pink!70!white}{.}}}

\theoremstyle{plain}
\newtheorem{theorem}{Theorem}
\newtheorem{proposition}[theorem]{Proposition}

\newtheorem{definition}[theorem]{Definition}

\newtheorem{claim}[theorem]{Claim}
\theoremstyle{remark}
\newtheorem{remark}{Remark}
\newtheorem{example}{Example}

\usepackage{pifont}

\newcommand{\commentout}[1]{}

\newcommand{\la}{\langle}
\newcommand{\ra}{\rangle}

\newcommand{\NN}{\mathbb{N}}

\newcommand{\tuple}[1]{\langle #1 \rangle}
\newcommand{\tuplevar}[1]{\{ #1 \}}

\newcommand{\hl}{\mathfrak{L}}
\newcommand{\hlang}{\mathfrak{L}}
\newcommand{\lang}{\mathcal{L}}
\newcommand{\A}{{\mathcal{A}}}

\newcommand{\wass}[2]{{\bi #1}{_{#2}}}
\newcommand{\assw}[2]{{#1}_{\bi #2}}

\newcommand{\naturals}{\NN}

\newcommand{\comp}[1]{\textsf{\small #1}}

\newcommand{\bi}[1]{\textbf{\textit #1}} 

\newcommand{\reglang}[1]{\mathcal{L}(#1)}
\newcommand{\Le}[1]{\mathbf{L}(#1)}
\newcommand{\Ri}[1]{\mathbf{R}(#1)}

\title{Realizable and Context-Free Hyperlanguages}
\author{Hadar Frenkel
\institute{CISPA Helmholtz Center for Information Security, Saarbrücken, Germany }
\email{hadar.frenkel@cispa.de}
\and
Sarai Sheinvald 
\institute{Department of Software Engineering, Braude College of Engineering, Karmiel, Israel}
\email{sarai@braude.ac.il}
}
\def\titlerunning{Realizable and Context-Free Hyperlanguages}
\def\authorrunning{H. Frenkel \& S. Sheinvald}
\begin{document}
\maketitle

\begin{abstract}
{\em Hyperproperties} lift conventional trace-based languages from a set of execution traces to a
set of sets of executions. From a formal-language perspective, these are sets of sets of words, namely {\em hyperlanguages}. {\em Hyperautomata} are based on classical automata models that are lifted to handle hyperlanguages. {\em Finite hyperautomata} (NFH) have been suggested to express regular hyperproperties. We study the {\em realizability} problem for regular hyperlanguages: given a set of languages, can it be precisely described by an NFH? We show that the problem is complex already for singleton hyperlanguages. 
We then go beyond regular hyperlanguages, and study {\em context-free} hyperlanguages. We show that the natural extension to context-free hypergrammars is highly undecidable. We then suggest a refined model, namely {\em synchronous hypergrammars}, which enables describing interesting non-regular hyperproperties, while retaining many decidable properties of context-free grammars.  
\end{abstract}
\section{Introduction}\label{sec:intro}

{\em Hyperproperties}~\cite{cs10} generalize traditional trace properties~\cite{as85} to {\em system properties}, i.e., from sets of traces to sets of sets of traces. A hyperproperty dictates how a system should behave in its entirety and not just based on its individual executions. 
Hyperproperties have been shown to be a powerful tool for expressing and reasoning about 
information-flow security policies~\cite{cs10} and important properties of cyber-physical 
systems~\cite{wzbp19} such as {\em sensitivity} and {\em robustness}, as well as consistency 
conditions in distributed computing such as
{\em linearizability}~\cite{bss18}. 
Different types of logics, such as HyperLTL, HyperCTL$^*$~\cite{DBLP:conf/post/ClarksonFKMRS14},
HyperQPTL~\cite{DBLP:phd/dnb/Rabe16} and HyperQCTL$^*$~\cite{DBLP:conf/lics/CoenenFHH19}
have been suggested for expressing hyperproperties.

In the {\em automata-theoretic approach} 
both the system and the specification are modeled as automata whose language is the sets of execution traces of the system, and the set of executions that satisfy the specification ~\cite{DBLP:conf/lics/VardiW86,DBLP:conf/banff/Vardi95}. Then, problems such as model-checking~\cite{DBLP:books/daglib/0007403-2} (``Does the system satisfy the property?'') and satisfiability (``Is there a system that satisfies the property?'') are reduced to decision problems for automata, such as containment (``Is the language of an automaton $A$ contained in the language of automaton $B$?'')  and nonemptiness (``Is there a word that the automaton accepts?''). 
Finite-word and $\omega$-regular automata are used for modeling trace specifications~\cite{vw94}. {\em Hyperautomata}, introduced in~\cite{bs21}, generalize word automata to automata that run on sets of words. Just as hyperproperties describe a system in its entirety, the {\em hyperlanguages} of hyperautomata describe the language in its entirety. 

The work in~\cite{bs21} focuses on {\em nondeterministic finite-word hyperautomata} (NFH), that are able to model {\em regular hyperlanguages}. 
An NFH $\A$ uses {\em word variables} that are assigned words from a language $\lang$, as well as a {\em quantification condition} over the variables, which describes the existential ($\exists$) and global ($\forall$) requirements from $\lang$. 
The {\em underlying NFA} of $\A$ runs on the set of words assigned to the word variables, all at the same time. The hyperlanguage of $\A$ is then the set of all languages that satisfy the quantification condition. The decidability of the different decision problems for NFH heavily depends on the quantification condition. For example, nonemptiness of NFH is decidable for the conditions $\forall^*$ (a sequence of $\forall$-quantifiers), $\exists^*$ and $\exists^*\forall^*$, but is undecidable for $\forall\exists$. 
In~\cite{ggs21} NFH are used to specify {\em multi-properties}, which express the behaviour of several models that run in parallel. 

A natural problem for a model $M$ for languages is {\em realizability}: given a language $\lang$, can it be described by $M$? For finite-word automata, for example, the answer relies on the number of equivalence classes of the Myhill-Nerode relation for $\lang$. 
For hyperlanguages, we ask whether we can formulate a hyperproperty that precisely describes a given set of languages. We study this problem for NFH: given a set $\hl$ of languages, can we construct an NFH whose hyperlanguage is $\hl$? We can ask the question generally, or for specific quantification conditions.
We focus on a simple case of this problem, where $\hl$ consists of a single language, (that is, $\hl=\{\lang\}$ for some language $\lang$), which turns out to be non-trivial. 
In~\cite{DBLP:conf/csfw/FinkbeinerHT19}, the authors present automata constructions for safety regular hyperproperties. As such, they are strictly restricted only to $\forall^*$-conditions.

We show that for the simplest quantification conditions, $\exists^*$ and $\forall^*$, no singleton hyperlanguage is realizable, and that single alternation does not suffice for a singleton hyperlanguage consisting of an infinite language.
We show that when $\lang$ is finite, then $\{\lang\}$ is realizable with a $\forall\exists$ quantification condition.

We then define {\em ordered} languages. These are languages that can be enumerated by a function $f$ that can be described by an automaton that reads pairs of words: a word $w$, and $f(w)$. We show that for an ordered language $\lang$, the hyperlanguage $\{\lang\}$ is realizable with a quantification condition of $\exists\forall\exists$.
We then generalize this notion to {\em partially ordered} languages, which are enumerated by a relation rather than a function. We show that for this case, $\{\lang\}$ is realizable with a quantification condition of $\exists^*\forall\exists^*$. 

Finally, we use ordered languages to realize singleton hyperlanguages consisting of regular languages: we show that when $\lang$ is a prefix-closed regular language, then it is partially ordered, and that an NFH construction for $\{\lang\}$ is polynomial in the size of a finite automaton for $\lang$. 
We then show that every regular language is partially ordered. Therefore, when $\lang$ is regular, then $\{\lang\}$ is realizable with a quantification condition of $\exists^*\forall\exists^*$. 
The summary of our results is listed in Table~\ref{tab:realizability.summary}.


\begin{table}[t]
    \centering
    \begin{tabular}{|c|c|c|} 
    \hline
         {Type of language} & Quantification Type & Realizability
         \\ 
         \hline
         \multirow{2}{*} {Finite }& $\exists^*,~\forall^*$ & unrealizable (T.\ref{thm:unrealizable}) \\
         & $\forall\exists$ & polynomial (T.\ref{thm:finite.fe.realizable})
        \\
 \hline
 Infinite & $\exists^*,~\forall^*,~\exists^*\forall^*$ & unrealizable (T.\ref{thm:unrealizable})\\
 \hline
 Ordered & $\exists\forall\exists$ & polynomial (T. \ref{thm:ordered.realizable})\\
 \hline
 partially Ordered &$\exists^*\forall\exists^*$ & exponential (T.\ref{thm:partial.order.realizable})\\
 \hline
 Prefix-Closed Regular &$\exists\forall\exists^*$ & polynomial (T.\ref{thm:prefix.closed.realizable}) \\
 \hline
 Regular &$\exists^*\forall\exists^*$ & doubly exponential (T.\ref{thm:regular.realizable})\tablefootnote{See remark~\ref{rem:automatic}.} \\
 \hline
 
    \end{tabular}
    \caption{Summary of realizability results for singleton hyperlanguages.}
    \label{tab:realizability.summary}
\vspace{-4mm}
\end{table}

In the second part of the paper, we go beyond regular hyperlanguages, and study {\em context-free hyperlanguages}. To model this class, we generalize context-free grammars (CFG) to {\em context-free hypergrammars} (CFHG), similarly to the generalization of finite-word automata to NFH: we use an {\em underlying CFG} that derives the sets of words that are assigned to the word variables in the CFHG $G$. The quantification condition of $G$ defines the existential and global requirements from these assignments. 

The motivation for context-free hyperlanguages is clear: they allow expressing more interesting hyperproperties. As a simple example, consider a robot which we want to return to its charging area before its battery is empty. This can be easily expressed with a CFHG with a $\forall$-condition, as we demonstrate in Section~\ref{sec:CFGH}, Example~\ref{ex:CFHG}. 
Note that since the underlying property -- {\em charging time is larger than action time} -- is non-regular, NFH cannot capture this specification. Extending this example, using an $\exists\forall$-condition we can express the property that all such executions of the robot are bounded, so that the robot cannot charge and act unboundedly.

Some aspects regarding context-free languages in the context of model-checking have been studied. In~\cite{DBLP:conf/stacs/Finkbeiner017,DBLP:conf/spin/PommelletT18} the authors explore model-checking of HyperLTL properties with respect to context-free models. There, the systems are context-free, but not the specifications. The work of~\cite{DBLP:conf/concur/BouajjaniER94} studies the verification of non-regular temporal properties, and~\cite{DBLP:journals/acta/FridmanP14} studies the synthesis problem of context-free specifications. These do not handle context-free \emph{hyperproperties}.

While most natural decision problems are decidable for regular languages, this is not the case for CFG. For example, the universality (``Does the CFG derive all possible words?'') and containment problems for context-free languages are undecidable. 
The same therefore holds also for CFHG. However, the nonemptiness and membership (``Does the CFG derive the word $w$''?) problems are decidable for CFG. 

We study the various decision problems for CFHG.
Specifically, We explore the nonemptiness problem (``Is there a language that the CFHG derives?''); and the membership problem (``Does the CFHG derive the language $\lang$''?). These problems correspond to the satisfiability and model-checking problems, respecitively. 
We show that for general CFHG, most of these problems soon become undecidable (see Table~\ref{tab:cfg.summary}). Some of the undecidability results are inherent to CFG. Some, however, are due to the {\em asynchronous} nature of CFHG: when the underlying CFG of a CFHG derives a set of words that are assigned to the word variables, it does not necessarily do so synchronously. For example, in one derivation step one word in the set may be added $2$ letters, and another $1$ letter. NFH read one letter at a time from every word in the set, and are hence naturally synchronous. In \cite{bcbfs21}, the authors study asynchronous hyperLTL, which suffers from the same phenomenon. 

We therefore define {\em synchronous context-free hyperlanguages}, which require synchronous reading of the set of words assigned to the word variables. We also define {\em synchronous CFHG} (syncCFHG), a fragment of CFHG in which the structure of the underlying CFG is limited in a way that ensures synchronous behavior. We prove that syncCFHG precisely captures the class of synchronous context-free hyperlanguages. Further, we show that some of the undecidable problems for CFHG, such as the nonemptiness problem for the $\forall^*$- and $\exists\forall^*$-fragments, become decidable for syncCFHG.


\begin{table}[t]
    \centering
    \begin{tabular}{|c | c|c|c|} \hline
         \multicolumn{2}{|c|}{Problem} &  syncCFHGs & General CFHGs 
         \\ \hline
         \multirow{4}{*} {Emptiness }& $\exists^*$ & polynomial (T.\ref{thm:ExistsRankedEmptiness}) & polynomial (T.\ref{thm:ExistsRankedEmptiness})  \\
         & $\forall^* $,  $\exists\forall^*$  & polynomial (T.\ref{thm:forallsync}) & undecidable (T.\ref{thm:undecforall}) \\
         & $\exists^*\forall^*$ & undecidable (T.\ref{thm:emptinessexistsforall}) & undecidable (T.\ref{thm:undecforall}) 
        \\
 
  \hline
 Finite Membership& $*$ & 
 exponential 
 (R.\ref{rem:finitemem})&  exponential (T.\ref{thm:finitemembership})
          \\
 
 \hline
 \multirow{2}{*}{Regular Membership} &$ \exists^*$&  exponential (R.\ref{rem:regmemsync})
 &  exponential~(T.\ref{thm:regmodelcheck}) \\ 
 & $\forall^* $ & undecidable (T.\ref{thm:forallsyncundec}) & undecidable (T.\ref{thm:forallsyncundec})
 
         \\

 \hline
    \end{tabular}
    \caption{Summary of decidability results for synchronous and general CFHGs.}
    \label{tab:cfg.summary}
    \vspace{-4mm}
\end{table}

\section{Preliminaries}\label{sec:prelim}

\paragraph{Hyperautomata}

We assume that the reader is familiar with the definitions of deterministic finite automata (DFA) and non-deterministic finite automata (NFA). 
\begin{definition}
\label{def:nfh}
Let $\Sigma$ be an alphabet.
A {\em hyperlanguage $\hl$ over $\Sigma$} is a set of languages over $\Sigma$, that is, $\hl\in 2^{2^{\Sigma^*}}$. A {\em nondeterministic finite-word hyperautomaton} (NFH) is a tuple 
$\A = \tuple{\Sigma,X,Q,Q_0,F,\delta,\alpha}$, where $X$ is a finite set of {\em word variables}, and 
$\alpha = \quant_1 x_1\  \cdots \quant_k x_k$ is a {\em quantification condition}, where $\quant_i\in \{\exists,\forall\}$ for 
every $i\in[1,k]$, and $\tuple{{\hat\Sigma},Q,Q_0,F,\delta}$ forms an underlying NFA over $\hat\Sigma = (\Sigma\cup\{\#\})^X$. 
\end{definition}
Let $\lang$ be a language. 
We represent an assignment $v:X\rightarrow \lang$ as a {\em word assignment} $\wass{w}{v}$, which is a word over the alphabet $(\Sigma\cup\{\#\})^X$ (that is, assignments from $X$ to $(\Sigma\cup\{\#\})^*$), where the $i$'th letter of $\wass{w}{v}$ represents the $k$ $i$'th letters of the words 
$v(x_1),\ldots ,v(x_k)$ (in case that the words are not of equal length, we ``pad''  the end of the shorter words with $\#$-symbols).
We represent these $k$ $i$'th letters as an assignment 
denoted $\{\sigma_{1{x_1}},\sigma_{2{x_2}},\ldots, \sigma_{k{x_k}}\}$, where $x_j$ is assigned $\sigma_j$.
For example, the assignment $v(x_1)=aa$ and $v(x_2)=abb$ is represented by the word assignment $\wass{w}{v} =  \{a_{x_1},a_{x_2}\}\{a_{x_1},b_{x_2}\}\{\#_{x_1},b_{x_2}\}$. 

The acceptance condition for NFH is defined with respect to a language $\lang$, the underlying NFA $\hat\A$, the quantification condition $\alpha$, and an assignment $v:X\rightarrow \lang$. 

\begin{itemize}
    \item For $\alpha = \epsilon$, define $\lang\vdash _v (\alpha,\hat\A)$ if $\wass{w}{v}\in\lang(\hat\A)$. 

\item For $\alpha = \exists x. \alpha'$, define $\lang\vdash_v (\alpha, \hat\A)$ if  there exists $w\in \lang$ s.t. $\lang \vdash_{v[x\mapsto w]} (\alpha',\hat\A)$.

\item For $\alpha = \forall x. \alpha'$, define $\lang\vdash_v (\alpha,\hat\A)$ if $\lang \vdash_{v[x \mapsto w]}  
(\alpha',\hat A)$ for every $w\in \lang$ .\footnote{In case that $\alpha$ begins with $\forall$, membership holds vacuously with the empty language. We 
restrict the discussion to satisfaction by nonempty languages.}
\end{itemize}

When $\alpha$ includes all of $X$, then membership is independent of the assignment, and we say that {\em $\A$ accepts $\lang$}, and denote $\lang\in\hlang(\A)$.

\begin{definition}
Let $\A$ be an NFH. The {\em hyperlanguage} of $\A$, denoted $\hlang{(\A)}$, is 
the set of all languages that $\A$ accepts.
When the quantification condition $\alpha$ of an NFH $\A$ is $\quant_1 x_1 .\quant 2 x_2 \cdots \quant_k x_k$, we denote $\A$ as being a $\quant_1\quant_2 \dots \quant_k$-NFH (or, sometimes, as an $\alpha$-NFH). 
\end{definition}

\begin{example}
Consider the NFH $\A$ depicted in Figure~\ref{fig:nfh}, over the alphabet $\{a\}$. The quantification condition $\forall x.\exists y$ requires that in a language $\lang$ accepted by $\A$, for every word $u_1$ that is assigned to $x$, there exists a word $u_2$ that is assigned to $y$ such that the joint run of $u_1,u_2$ is accepted by the underlying NFA $\hat\A$ of $\A$. 
The NFA $\hat\A$ requires that the word assigned to $y$ is longer than the word assigned to $x$: once the word assigned to $x$ ends (and the padding $\#$ begins), the word assigned to $y$ must still read at least one more $a$. Therefore, $\A$ requests that for every word in $\lang$, there exists a longer word in $\lang$. This holds iff $\lang$ is infinite. Therefore, the hyperlanguage of $\A$ is the set of all infinite languages over $\{a\}$.
\end{example}

\begin{figure}[t]
\centering
\scalebox{.8}{
        \includegraphics[scale=0.8]{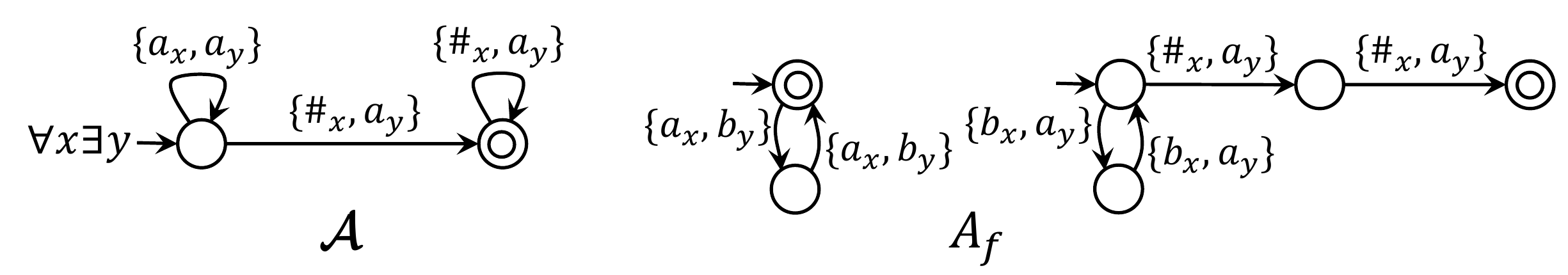}
    }
   \caption{The NFH $\A$ (left), whose hyperlanguage is the set of infinite languages over $\{a\}$, and the NFA $A_f$ that computes $f$ (right).}
    \label{fig:nfh}
    \vspace{-3mm}
\end{figure}

\paragraph{Context-Free Grammars}
\begin{definition}
A {\em context-free grammar} (CFG) is a tuple $G= \la\Sigma,V,V_0,P\ra$, where $\Sigma$ is an alphabet, $V$ is a set of {\em grammar variables}, $V_0\in V$ is an {\em initial variable}, and $P\subseteq V\times (V\cup\Sigma)^*$ is a set of {\em grammar rules}.
\end{definition}
We say that $w$ is a \emph{terminal word} if $w\in \Sigma^*$. 
Let $v\in V$ and $\alpha, \beta\in (V\cup\Sigma)^*$.
\begin{itemize} \setlength{\itemsep}{0pt}
  \setlength{\parskip}{1pt}
    \item We say that $v$ \emph{derives} $\alpha$
if $(v,\alpha) \in P$. We then denote $v\rightarrow \alpha$.
    \item We denote $\alpha \Rightarrow \beta$ if there exist $v\in V$ and $\alpha_1, \alpha_2, \beta' \in (V\cup\Sigma )^*$ such that $v\rightarrow\beta'$, $\alpha = \alpha_1 v\alpha_2$ and $\beta = \alpha_1 \beta' \alpha_2$. 
    \item We say that $\alpha$ \emph{derives} $\beta$, or that $\beta$ is \emph{derived} by $\alpha$, if there exists $n\in \mathbf{N}$ and $\alpha_1, \ldots \alpha_n \in (V\cup \Sigma)^*$ such that $\alpha_1 = \alpha$, $\alpha_n = \beta$ and $\forall 1\leq i < n:~ \alpha_i\Rightarrow \alpha_{i+1}$. We then denote $\alpha \Rightarrow^* \beta$.
\end{itemize}

The \emph{language} of a CFG $G$ is the set of all terminal words that are derived by the initial variable. That is, $\mathcal{L}(G) = \{ w\in \Sigma^* ~|~ V_0\Rightarrow^* w \}$. 

\section{Realizability of Regular Hyperlanguages} \label{sec:realizability}

Every NFH $\A$ defines a set of languages $\hl$. In the {\em realizability problem} for NFH, we are given a hyperlanguage $\hl$, and ask whether there exists an NFH $\A$ such that $\hl(\A)=\hl$. The answer may depend on the quantification condition that we allow using. In this section we study the realizability problem for singleton hyperlanguages, which turns out to be non-trivial. We show that while for finite languages we can construct a $\forall\exists$-NFH, such a condition cannot suffice for infinite languages. Further, we show that a general regular language requires a complex construction and quantification condition.  

\begin{definition}\label{realizability}
Let $\hl$ be a hyperlanguage. 
For a sequence of quantifiers $\alpha = \quant_1\dots\quant_k$, we say that $\hl$ is {\em $\alpha$-realizable} if there exists an NFH $\A$ with a quantification condition $\quant_1 x_1\dots \quant_k x_k$ such that $\hl(\A) = \hl$. 
\end{definition}

We first define some operations and notations on the underlying NFA of NFH we use in our proofs.

For a word $w$, the NFA $A_w$ is an NFA for $\{w\}$.


Let $A_1,A_2$ be NFA, and let  $A_1\uparrow^\#=\tuple{\Sigma,Q,q_0,\delta_1,F_1}$ and $A_2\uparrow^\#=\tuple{\Sigma,P,p_0,\delta_2,F_2}$, where $A_i\uparrow^\#$ is an NFA for the language $\lang(A_i)\cdot{\#^*}$. We define the {\em composition} $A_1\otimes A_2$ of $A_1$ and $A_2$ to be an NFA over $\Sigma_2 =  (\Sigma\cup\{\#\})^{\{x,y\}}$, defined as
$A_1\otimes A_2 = \tuple{\Sigma_2,Q\times P, (q_0,p_0),\delta, F_1\times F_2}$, where for every $(q,\sigma,q')\in \delta_1$, $(p,\tau,p')\in \delta_2$, we have $((q,p),\{\sigma_x,\tau_y\},(q',p'))\in \delta$. That is, $A_1\otimes A_2$ is the composition of $A_1$ and $A_2$, which follows both automata simultaneously on two words (adding padding by $\#$ when necessary). A word assignment $\{x\mapsto w_1,y\mapsto w_2\}$ is accepted by $A_1\otimes A_2$ iff $w_1\in\lang(A_1)$ and $w_2\in\lang(A_2)$ (excluding the $\#$-padding). 

We also define $A_1\oplus A_2$ of $A_1$ and $A_2$ in a similar way, but the transitions are restricted to equally labeled letters. That is, for every $(q,\sigma,q')\in \delta_1$, $(p,\sigma,p')\in \delta_2$, we have $((q,p),\{\sigma_x,\sigma_y\},(q',p'))\in \delta$.
A run of the restricted composition then describes the run of $A_1$ and $A_2$ on the same word. 

We generalize the definition of both types of compositions to a sequence of $k$ NFA $A_1,A_2,\ldots A_k$, forming NFA $\bigotimes_{i=1}^k A_i$ and $\bigoplus_{=1}^k A_i$ over $(\Sigma\cup{\#})^{\{x_1,x_2,\ldots x_k\}}$, in the natural way.
When all NFA are equal to $A$, we denote this composition by $A^{\otimes k}$ (or $A^{\oplus k}$). 
When we want to explicitly name the variables $x_1,x_2,\ldots x_k$ in the compositions, we denote $\bigotimes_{i=1}^k A_i[x_1,\ldots x_k]$ (or $\bigoplus_{i=1}^k A_i[x_1,\ldots x_k]$). 

We also generalize the notion of composition to NFA over $(\Sigma\cup\{\#\})^X$ in the natural way. That is, for NFA $A_1$ and $A_2$ with sets of variables $X$ and $Y$, respectively, the NFA $A_1\otimes A_2$ is over $X\cup Y$, and follows both NFA simultaneously on both assignments (if  $X\cap Y \neq \emptyset$, we rename the variables). 

\stam{
Let 
$A_1=\tuple{(\Sigma\cup\{\#\})^X,Q,q_0,\delta_1,F_1}$ and $A_2\uparrow^\#=\tuple{\Sigma,P,p_0,\delta_2,F_2}$
We denote by $A_1\cap_x A_2$ the NFA over $(\Sigma\cup\{\#\})^X$ that is formed by the product construction of $A_1$ and $A_2$ with respect to $x$.
That is, for every $(q,S\cup\{\sigma_x\},q')\in \delta_1$ and $(p,\sigma,p')\in \delta_2$, we have $((q,p),S\cup\{\sigma_x\},(q',p'))$. The set of accepting states is $F_1\times F_2$.
}

\bigskip

We study the realizability problem for the case of singleton hyperlanguages, that is, hyperlanguages of the type $\{\lang\}$. We begin with a few observations on unrealizability of this problem, and show that general singleton hyperlanguages cannot be realized using simple quantification conditions. 

\subsection{Unrealizability}
For the homogeneous quantification conditions, we have that a $\forall$-NFH $\A$ accepts a language $\lang$ iff $\A$ accepts every $\lang'\subseteq \lang$. Therefore, a hyperlanguage $\{\lang\}$ is not $\forall$-realizable for every $\lang$ that is not a singleton.  
The same holds for every $\forall^*$-NFH.

An $\exists$-NFH $\A$ accepts a language $\lang$ iff
$\lang$ contains some word that is accepted by $\hat\A$. Thus if $\A$ is nonempty, its hyperlanguage is infinite, and clearly not a singleton.
The same holds for every $\exists^*$-NFH.

Now, consider an $\exists^k\forall^m$-NFH $\A$. As shown in \cite{bs21}, $\A$ is nonempty iff it accepts a language whose size is at most $k$. Therefore, $\{\lang\}$ is not $\exists^k\forall^m$-realizable for every $\lang$ such that $|\lang|>k$.

As we show in Theorem~\ref{thm:finite.fe.realizable}, if $\lang$ is finite then $\{\lang\}$ is $\forall\exists$-realizable. We now show that if $\lang$ is infinite, then $\{\lang\}$ is not $\forall\exists$-realizable. 
Assume otherwise by contradiction, and let $\A$ be a $\forall x.\exists y$-NFH that accepts $\{\lang\}$. Then for every $w\in \lang$ there exists $u\in \lang$ such that $\wass{w}{[x\mapsto w][y\mapsto u]}\in \lang(\hat\A)$. Let $w_1$ be some word in $\lang$. 
We construct an infinite sequence $w_1,w_2,\ldots$ of words in $\lang$, as follows. For every $w_i$, let $w_{i+1}$ be a word in $\lang$ such that $\wass{w}{[x\mapsto w_i][y\mapsto w_{i+1}]}\in \lang(\hat\A)$. 
If $w_i = w_j$ for some $i<j$, then the language $\{w_i,w_{i+1},\ldots,w_j\}$ is accepted by $\A$, and so $\lang$ is not the only language that $\A$ accepts. 
Otherwise, all words in the sequence are distinct. Then, the language $\lang_i=\{w_i,w_{i+1},\ldots \}$ is accepted by $\A$ for every $i>1$, and $\lang_i\subset \lang$. In both cases, $\lang$ is not the only language that $\A$ accepts, and so $\{\lang\}$ is not $\forall\exists$-realizable. 

To conclude, we have the following. 
\begin{theorem}\label{thm:unrealizable}
If $\lang$ contains more than one word, then $\{\lang\}$ is not $\forall^*$-realizable and not $\exists^*$-realizable. If $\lang$ contains more than $k$ words then $\{\lang\}$ is not $\exists^k\forall^*$-realizable. 
If $\lang$ is infinite then 
$\{\lang\}$ is not $\exists^*\forall^*$-realizable and not $\forall\exists$-realizable.
\end{theorem}

For positive realizability results, we first consider a simple case of a hyperlanguage consisting of a single finite language.

\begin{theorem}\label{thm:finite.fe.realizable}
Let $\lang$ be a finite language. Then $\{\lang\}$ is $\forall\exists$-realizable.
\end{theorem}

\begin{proof}
Let $\hlang=\{w_1,w_2,\ldots w_k\}$. We construct a $\forall\exists$-NFH $\A$ for $\lang$, whose underlying NFA $\hat\A$ is the union of all NFA $A_{w_i}\otimes A_{w_{i+1}(\textrm{mod}~ k)}$. 
Let $\lang'\in \hl(\A)$. Since $\hat\A$ can only accept words in $\lang$, we have that $\lang'\subseteq \lang$. 
Since $\A$ requires, for every $w_i\in \lang'$, the existence of $w_{i+1(\textrm{mod}~ k)}$, and since $\lang'\neq \emptyset$, we have that by induction, $w_i\in \lang'$ implies $w_{i+j (\textrm{mod}~ k)}\in \lang'$ for every $1\leq j\leq k$. Therefore, $\lang\subseteq \lang'$.  
\end{proof}

\subsection{Realizability of Ordered Languages}

Since every language is countable, we can always order its words. We show that for an ordering of a language $\lang$ that is {\em regular}, that is, can be computed by an NFA, $\{\lang\}$ can be realized by an $\exists\forall\exists$-NFH.  

\begin{definition}
Let $\lang$ be a language. We say that a function $f:\lang\rightarrow \lang$ is {\em $\lang$-regular} if there exists an NFA $A_f$ over $(\Sigma\cup\{\#\})^{\{x,y\}}$ such that for every $w\in \lang$, it holds that $f(w)=u$ iff  $\wass{w}{[x\mapsto w][y\mapsto u]}\in \lang(A_f)$. We then say that $A_f$ {\em computes} $f$. 

We say that a language $\lang$ is {\em ordered} if the words in $\lang$ can be arranged in a sequence $w_1, w_2,\ldots$ such that there exists an $\lang$-regular function such that $f(w_i)=w_{i+1}$ for every $i\geq 0$.
\end{definition}

\begin{example}
Consider the language $\{a^{2i},b^{2i}|i\in\naturals\}$, and a function $\forall i\in\naturals: f(a^{2i})=b^{2i}, f(b^{2i})=a^{2i+2}$, which matches the sequence $\varepsilon, a^2,b^2,a^4,b^4,\ldots$. The function $f$ can be computed by the NFA $A_f$ depicted in Figure~\ref{fig:nfh}, which has two components: one that reads $a^{2i}$ on $x$ and $b^{2i}$ on $y$,
and one that reads $b^{2i}$ on $x$ and $a^{2i+2}$ on $y$. 
\end{example}

\begin{theorem}\label{thm:ordered.realizable}
Let $\lang$ be an ordered language. Then $\{\lang\}$ is $\exists\forall\exists$-realizable. 
\end{theorem}

\begin{proof} 
Let $\lang = \{w_1,w_2,\ldots\}$ be an ordered language via a regular function $f$, and let $A_f$ be an NFA that computes $f$. 
We construct an $\exists x_1 \forall x_2 \exists x_3$-NFH $\A$ for $\{\lang\}$ by setting its underlying NFA to be $\hat\A = A_{w_1}\otimes A_f[x_1,x_2,x_3]$. 
Intuitively, $\A$ creates a ``chain-reaction'':  $A_{w_1}$ requires the existence of $w_1$, and $A_f$ requires the existence of $w_{i+1}$ for every $w_i$. By the definition of $f$, only words in $\lang$ may be assigned to $x_2$. Therefore, $\hlang(\A) = \{\lang\}$.  
\stam{
\sarai{the rest is straightforward, move to appendix}
According its definition, $\hat\A$ only accepts word assignments in which $x_1,x_2$ and $x_3$ are assigned (padded) words in $\lang$. Therefore, $\lang'\subseteq \lang$ for every $\lang'\in\hl(\A)$. 

Now, let $\lang'\in\hl(\A)$
We prove by induction that $w_i\in \lang$ for every $i\geq 1$.
Consider a word $\wass{w}{v}\in \lang(\hat\A)$. Then $v(x_1)=w_1$. Since $x_1$ is under $\exists$, we have that $w_1\in \lang'$. 
Assume that $w_i\in \lang'$. Then, since $x_2$ is under $\forall$, for $\lang'$ to be accepted by $\A$ there must exist $\wass{w}{v}\in \lang(\hat\A)$ such that $v(x_2)=w_i$. Since $A_f$ computes $f$, we have $v(x_3)=w_{i+1}$. Since $x_3$ is under $\exists$, then since $\lang'$ is accepted, it must include $w_{i+1}$. Therefore $\lang\subseteq\lang'$, and so $\lang=\lang'$, and we have $\hlang(\A)=\{\lang\}$. 
}
\end{proof}

We now generalize the definition of ordered languages, by allowing several minimal words instead of one, and allowing each word to have several successors. The computation of such a language then matches a {\em relation} over the words, rather than a function. 

\begin{definition}
We say that a language $\lang$ is $m,k$-ordered, if there exists a relation $R\subseteq \lang\times\lang$ such that:
\begin{itemize}
    \item There exist exactly $m$ words $w\in\lang$ such that $(u,w)\notin R$ for every $u\neq w\in\lang$ (that is, there are $m$ minimal words).
    \item $R\subseteq S$ for a total order $S$ of $\lang$ with a minimal element. 
    \item For every $w\in\lang$ there exist $1\leq i\leq k$ successor words: words $u$ such that $(w,u)\in R$.  
    \item There exists an NFA $A_R$ over $(\Sigma\cup\{\#\})^{\{x,y\}}$ such that for every $u,v\in\lang$, it holds that $R(u,v)$ iff $\wass{w}{[x\mapsto u][y\mapsto v]}\in\lang(A_R)$.  
\end{itemize}
We then say that $A_R$ {\em computes $\lang$}.
We call $\lang$ {\em partially ordered}
if there exist $m,k$ such that $\lang$ is $m,k$-ordered.
\end{definition}


\begin{theorem}\label{thm:partial.order.realizable}
Let $\lang$ be $m,k$-ordered. Then $\{\lang\}$ is $\exists^m\forall\exists^k$-realizable. 
\end{theorem}
\begin{proof}
Let $A_R = \tuple{(\Sigma\cup\{\#\})^{\{x,y\}},Q,q_0,\delta,F}$ be a DFA that computes $\lang$. 
We construct an $\exists^m\forall\exists^k$-NFH $\A$ for $\lang$, as follows. 
The quantification condition of $\A$ is $\exists x_1 \cdots \exists x_m \forall z \exists y_1 \cdots \exists y_k$.
The $x$-variables are to be assigned $u_1,\ldots u_m$, the $m$ minimal words of $R$.
We set $A_U = \bigotimes_{i=1}^m A_{u_i}[x_1,\ldots x_m]$. 

The underlying NFA $\hat\A$ of $\A$ comprises of an NFA $A_i$ 
for every $1\leq i \leq k$. 
Let $\lang_i$ be the set of words in $\lang$ that have exactly $i$ successors. 
Intuitively, $A_i$ requires, for every word $w\in \lang_i$ that is assigned to $z$, the existence of the $i$ successors of $w$.

To do so, we construct an NFA $B_i$ over $z,y_1,\ldots y_i$ that accepts $\wass{w}{[z\mapsto w][y_1\mapsto w_1]\dots [y_i\mapsto w_i]}$, for every word $w$ and its successors $w_1,\ldots w_i$. 
The construction of $B_i$ requires that: (1) all assignment to $y$-variables are successors of the assignment to $z$, by basing $B_i$ on a composition of $A_R$, and (2)  $y_1,\ldots y_i$ are all assigned different words. This is done by keeping track of the pairs of assignments to $y$-variables that at some point read different letters. The run may accept only once all pairs are listed. We finally set $A_i = A_U\otimes B_i$, and require $y_{i+1}\ldots y_k$ to be equally assigned to $z$. 

Since $R\subseteq S$ and $R$ has minimal elements, for every word $w\in\lang$ there exists a sequence $w_0,w_1,\ldots w_t$ such that $w_t = w$, and $w_0$ is minimal, and $R(w_j,w_{j+1})$ for every $j\in[0,t-1]$. The NFH $\A$ then requires $w_0$ (via $A_u$), and for every $w_j$, requires the existence of all of its successors, according to their number, and in particular, the existence of $w_{j+1}$ (via $A_i$). Therefore, $\A$ requires the existence of $w_t$. On the other hand, by construction, every word that is assigned to $z$ is in $\lang$. Therefore, $\hl(\A) = \{\lang\}$.

The size of $A_U$ is linear in $|U|$, and the size of $A_i$ is exponential in $k$ and in $|A|$. 
\end{proof}

\subsection{Realizability of Regular Languages}

We now show that every regular language $\lang$ is  partially ordered. 
To present the idea more clearly, 
we begin with a simpler case of prefix-closed regular languages, and then proceed to general regular languages.
A prefix-closed regular language $\lang$ has a DFA $A$ in which every accepting state is only reachable by accepting states. We use the structure of $A$ to define a relation that partially orders $\lang$. 

\begin{theorem}\label{thm:prefix.closed.realizable}
Let $\lang$ be a non-empty prefix-closed regular language. Then $\lang$ is partially ordered. 
\end{theorem}
\begin{proof}
Let $A = \tuple{\Sigma,Q,q_0,\delta,F}$ be a DFA for $\lang$. Then for every $p,q\in Q$, if $q\in F$ and $q$ is reachable from $p$, then $p\in F$. 
Let $k$ be the maximal number of transitions from a state $q\in F$ to its neighboring accepting states.
We show that $\lang$ is $1,k$-ordered.
We define a relation $R$ as follows. 
Since $\lang$ is prefix-closed, we have that $\varepsilon\in\lang$. We set it to be the minimal element in $R$. 

Let $q\in Q$, and let $\{(q,\sigma_1,p_1),\ldots (q,\sigma_m,p_m)\}$ be the set of transitions in $\delta$ from $q$ to accepting states. 
For every word $w\in\lang$ that reaches $q$, we set $(w,w\sigma_1),\ldots (w,w\sigma_m)\in R$. 
For every word $w\in\lang$ that reaches a state $q$ from which there are no transitions to accepting states, we set $(w,w)\in R$. 

It holds that the number of successors for every $w\in\lang$ is between $1$ and $k$.
Further, $R\subseteq S$ for the length-lexicographic order $S$ of $\lang$. 

We construct an NFA $A_R$ for $R$ by replacing every transition labeled $\sigma$ with $\{\sigma_x, \sigma_y\}$ and adding a state $p'$, which is the only accepting state. For every $q\in F$, we add a transition $(q,\{\#_x,\sigma_y\},p')$ for every $(q,\sigma,p)\in\delta$ such that $p\in F$. If $q$ has no transitions to accepting states in $A$, then we add $(q,\{\#_x,\#_y\},p')$.
$A_R$ then runs on word assignments $\wass{w}{[x\mapsto w][y\mapsto u]}$ such that $u=w\sigma$ for some $\sigma$, such that $w,u\in \lang$, or $\wass{w}{[x\mapsto w][y\mapsto w]}$ if $w$ cannot be extended to a longer word in $\lang$. Therefore, $A_R$ computes $\lang$. 
\end{proof}
\begin{remark}
The construction in the proof of Theorem~\ref{thm:partial.order.realizable} is exponential, due to the composition of several automata. In the case of prefix-closed languages, the successors of a word $w\in \lang$ are all of the type $w\sigma$. Therefore, it suffices to extend every transition in $A$ to $\{\sigma_{x_1},\sigma_{x_2},\ldots \sigma_{x_k}\}$, and to add a transition from every $q\in F$ to a new accepting state with all letters leading from $q$ to an accepting state. Composed with a single-state DFA for $\epsilon$, we get an $\exists\forall\exists^k$-NFH for $\{\lang\}$, whose size is polynomial in $|A|$.  
\end{remark}

We now turn to prove the realizability of $\{\lang\}$ for every regular language $\lang$.
The proof relies on a similar technique to that of Theorem~\ref{thm:prefix.closed.realizable}: a relation that computes $\{\lang\}$ requires, for every word $w\in\lang$, the existence of a longer word $w'\in \lang$. Here, $w'$ is not simply the extension of $w$ by a single letter, but a pumping of $w$ by a single cycle in a DFA for $\lang$. 

\begin{theorem}\label{thm:regular.realizable}
Let $\lang$ be a regular language. Then $\{\lang\}$ is 
partially ordered. 
\end{theorem}

\begin{proof}
Let $A = \tuple{\Sigma,Q,q_0,\delta,F}$ be a DFA for $\lang$. We mark by $P$ the set of words that reach accepting states from $q_0$ along a  simple path. 
For a state $q\in Q$, 
we mark by $C_q$ the set of words that reach $q$ from $q$ along a simple cycle. Note that $P$ and $C_q$ are finite for every $q\in Q$. Let 
$n =|P|$, and let $m= \Sigma_{q\in Q}|C_q|$. 
We show that $\lang$ is $n,m$-ordered, by defining an appropriate relation $R$. 

The set of minimal words in $R$ is $P$. The successors of a word $w\in\lang$ are $w$ itself (that is, $R$ is reflexive), and every possible pumping of $w$ by a single simple cycle that precedes all other cycles within the run of $A$ on $w$. That is, for a state $q$ that is reached by a prefix $u$ of $w$ along a simple path, and for a word $c$ read along a simple cycle from $q$ to itself, the word $ucv$ is a successor of $w$ in $R$, where $w=uv$. 

To see that the only minimal words in $R$ are $P$,
let $w=\sigma_1\sigma_2\cdots\sigma_k\in\lang$, and let $r = (q_0,q_1,\ldots q_k)$ be the accepting run of $A$ on $w$. 
If all states in $r$ are unique, then $w\in P$. Otherwise, we set $w_t=w$, and repeatedly remove simple cycles from $r$: let $j$ be a minimal index for which there exists $j'>j$ such that $q_j=q_{j'}$ and such that $q_{j+1},\ldots q_{j'}$ are unique. We define $w_{i-1} = w_1\cdots w_j w_{j'+1}\cdots w_k$. We repeat this process until we reach a run in which all states are unique, which matches a word $w_0\in P$. The sequence of words $w_t, w_{t-1}, \ldots w_0$ we obtain is such that $(w_i,w_{i+1})\in R$ for every $i\in[0,t-1]$. 
\stam{
We show that there exist words $w_0,w_1, \ldots w_t$ such that $w_0\in P_{q_k}$ and $w_t=w$, and such that $(w_i,w_{i+1})\in R$. We start with $w_i=w_t$. If all states in $r$ are unique, then $w_0=w\in P_{q_k}$, and we are done. Otherwise, let $j<j'$ be such that $q_j=q_{j'}$, and such that $q_{j+1},\ldots q_{j'}$ are unique. We define $w_{i-1} = w_1\cdots w_j w_{j+1}\cdots w_k$. That is, we remove the subword between $q_j$ and $q_{j'}$. By the construction of $R$, we have $(w_{i-1},w_i)\in R$. We repeat the process until we reach a word for which all states are unique, and is therefore in $P_{q_k}$. 
}

It is easy to see that $R\subseteq S$ for the length-lexicographical order $S$ of $\lang$.
Additionally, every $w\in\lang$ has between $1$ and $m$ successors. 
We now construct an NFA $A_R$ for $R$. 

$A_R$ is the union of several components, described next. 
Let $A_q$ be the DFA obtained from $A$ by setting its only accepting state to be $q$.
For every $p\in Q$ and for every $c\in C_p$, we construct an NFA $B_{c,q}$, which pumps a word read along a run that reaches $q$ and traverses $p$, by $c$. The NFA $B_{c,q}$ comprises two copies $A_1,A_2$ of $A$, where the copy $q_2$ of $q$ in $A_2$ is the only accepting state. The word $c$ is read between $A_1$ to $A_2$, from $p_1$ and $p_2$. 

We construct an NFA $A_{c,q}$ by composing
$B_{c,q}$ and $A_q$, and making sure that $B_{c,q}$ reads the same word as $A_q$, pumped by $c$. That is, if $A_q$ reads a word $uv$, where $u$ reaches $p$, then $B_{c,q}$ reads $ucv$. To this end, while $B_{c,q}$ is in $A_1$, the DFA $A_q$ and $B_{c,q}$ both advance on the same letters. When $B_{c,q}$ leaves $A_1$ to read $c$ followed by the suffix $v$ in $A_2$, the composition remembers, via states, the previous (up to) $|c|$ letters read by $A_q$, to make sure that once $B_{c,q}$ finishes reading $uc$, it reads the same suffix $v$ as $A_q$ did.
The NFA $A_R$ is then the union of $A_{c,q}$ for every $q\in Q, c\in \bigcup_{p\in Q} C_p$.
To accept the reflexive pairs as well, we union all the components with an additional component $A\oplus A$. 

The size of every $A_{c,q}$ is exponential in $c$, due to the need to remember the previous $c$ letters. There are exponentially many simple paths and cycles in $A$. Therefore, we have that the size of $A_R$ is exponential in $|A|$. Combined with the exponential blow-up involved in the proof of Theorem~\ref{thm:partial.order.realizable}, we have that an NFH for $\{\lang\}$ is doubly-exponential in the $|A|$.   
\end{proof}

\begin{remark}\label{rem:automatic}
Using automatic structures~\cite{BN95} and relying on the length-lexicographical order $S$, one can prove the existence of an $\exists\forall\exists$-NFH $\A$ for $\{\lang\}$, which is smaller and simpler than the one we present in Theorem~\ref{thm:regular.realizable}. Indeed, one can phrase the direct successor relation in $\lang$ with respect to $S$ using the First Order Logic (FOL) formula $\varphi(x,y) = \lang(x)\wedge \lang(y)\wedge S(x,y) \wedge \forall(z). (z\neq y)\rightarrow (\neg(S(x,z)\wedge S(z,y)))$.
Since $S$ is NFA-realizable, and since every relation expressible by FOL over an automatic structure is regular~\cite{BN95}, we have that $\varphi$ is NFA-realizable. We can then construct $\A$, requiring the existence of a minimal word in $\lang$ with respect to $S$, together with the requirement of the existence of a successor for every $w\in\lang$. 

While this construction is polynomial, it does not directly rely on the structure of $A$. Since in this paper we wish to lay the ground for richer realizable fragments, in which relying on the underlying graph structures may be useful, we present it here. 
\end{remark}

\stam{
Let $\mathcal{L} = \{ L_i | i\in I \}$  be a hyperlanguage. When can we express $\mathcal{L}$ using a hyper automaton?

Informally: 
\begin{enumerate}
    \item  $\mathcal{L} = \{ L \}$ where $\lang = \{w_1, \ldots, w_k  \}$ a finite language, can be expressed using a $\forall\exists$ automaton with $k$ paths: $w_1, w_2$ to $w_k, w_1$. 
    \item  $\mathcal{L} = \{ L\}$ where $L$ is an infinite language, cannot be accepted by a $\forall\exists$ automaton (since if there is  a cycle of words  then also the finite language is accepted, and otherwise we can omit words and the language will still be accepted) 
    \item $\forall$ realizability - closed under subsets, $\exists$ realizability closed under supersets
    \item Let $ L \models \forall\exists\forall$ condition. Then $L \models \forall\exists$. Then, there exists a finite(???) $L' \subset L$ such that $L'\models \forall\exists$ (due to the cycle of words???) and thus $L ' \models \forall\exists\forall$. Thus, such a quantifying condition cannot accept any  single infinite language. 
    \item $\exists\forall\exists$: for ordered languages, we can union and concatenate. what about *?  
    \item given a DFA with a universal quantification, we get the hyper language of all of its subsets. with an existential quantification we get all intersections. What about all supersets? 
\end{enumerate}
}
\section{Context-Free Hypergrammars}\label{sec:CFGH}

We now go beyond regular hyperlanguages, and define and study {\em context-free hyperlanguages}. We begin with a natural definition for context-free hypergrammars (CFHG), based on the definition of NFH, and then identify a more decidable fragment of CFHG, namely {\em synchronized CFHG}.

\begin{definition}\label{def:CFHG}
A {\em context-free hypergrammar} (CFHG) is a tuple $\la\Sigma,X,V,V_0,P,\alpha\ra$, where $X$ and $\alpha$ are as in NFH, and where  $\hat{G}=\la{\hat\Sigma, V, V_0, P\ra}$ is a CFG over the alphabet $\hat\Sigma = (\Sigma\cup\{\#\})^X$.
\end{definition}
Definition~\ref{def:nfh} defines word assignments for NFH, where the $\#$-symbol may only appear at the end of a word. This is naturally enforced by the nature of the underlying NFA. 
For the most general case of hypergrammars, we consider words in which $\#$ can appear anywhere in the word. In Section~\ref{sec:sync} we allow $\#$ to occur only at the end of the word. For a word $w\in \Sigma^*$ we define the \emph{set of words} $w\uparrow_{\#}$ to be the set of all words that are obtained from $w$ by adding $\#$-symbols in arbitrary locations in $w$. For $w\in (\Sigma \cup \{\#\})^*$, we define the \emph{word} $w\downarrow_{\#}$ to be the word obtained from $w$ by removing all occurrences~of~$\#$. 
\stam{ 
Intuitively, a word $\hat w$ over $\hat\Sigma$ represents a set of words over $\Sigma$ as follows. Let $|X| = n$, and let $\hat{w} \in \hat{\Sigma}$ be $\hat{w} = \la w_1, \ldots, w_n  \ra $ such that $w_i\in \Sigma\cup \{ \#\}$. 
For each $w_i$, let $w_i'$ be the word that we obtained from $w_i$ be removing all occurrences of $\#$. Then, $\hat{w}$ represents the set $\{ w_1' , \ldots , w_n' \}$. Formally, we use the same word assignment $\wass{w}{u}$ given in Definition~\ref{def:nfh}, but we note that $\#$ can now appear anyway in the word. Then, we define the restriction of a word assignment  $\wass{w}{u}\downarrow_{\#}$ to be the word assignment that removes all occurrences of $\#$. We are now ready to define the language of a CFHG. 
\hadar{arrange notations here}
}

The acceptance condition for CFHG is defined with respect to a language $\lang$, the underlying CFG $\hat{G}$, the quantification condition $\alpha$, and an assignment $u:X\rightarrow \lang$. 

\begin{enumerate}
    \item For $\alpha = \epsilon$, define $\lang\vdash _v (\alpha,\hat{G})$ if $\wass{w}{u}\in\lang(\hat{G})$. \label{item:base}

\item For $\alpha = \exists x. \alpha'$, define $\lang\vdash_u (\alpha, \hat{G})$ if there exist $w\in \lang$ and $w_{\#}\in w\uparrow_{\#}$
s.t. $\lang \vdash_{u[x\mapsto w_{\#}]} (\alpha',\hat{G})$. \label{item:exists}

\item For $\alpha = \forall x. \alpha'$, define $\lang\vdash_u (\alpha,\hat{G})$ if
for every $w\in \lang$ there exists $w_{\#}\in w\uparrow_{\#}$ s.t.
$\lang \vdash_{u[x \mapsto w_{\#}]}  
(\alpha',\hat G)$.\label{item:forall}
\end{enumerate}

When $\alpha$ includes all of $X$, 
we say that {\em ${G}$ derives $\lang$} (or that ${G}$ accepts $\lang$), and denote $\lang\in\hlang({G})$.

\begin{definition}
Let ${G}$ be a CFHG. The {\em hyperlanguage} of ${G}$, denoted $\hlang{(G)}$, is 
the set of all languages that ${G}$ derives. We denote ${G}$ as being a $\quant_1\quant_2 \dots \quant_k$-CFHG similarly as with NFH.
\end{definition}


Henceforth we assume that (1) the underlying grammar $\hat{G}$ does not contain variables and rules that derive no terminal words (these can be removed); and (2) there are no rules of the form $v\rightarrow \varepsilon$
except for possibly $V_0 \rightarrow \varepsilon$. Every CFG can be converted to a CFG that satisfies these conditions~\cite{DBLP:books/daglib/Hopcroft}. 

 \begin{example}\label{ex:CFHG}
Consider the robot scenario described in Section~\ref{sec:intro}, and the $\forall x$-CFHG $G_1$ with the rules
\begin{align*}
   & P_1:= V_0\rightarrow \{ c_x\} V_0 \{ a_x\} ~|~\{  c_x \} V_1  \\
    & V_1\rightarrow \tuplevar{c_x} V_1  ~|~ \{c_x\}
\end{align*}
The letters $a$ and $c$ correspond to \emph{action} and \emph{charge}, respectively. Then, 
$\hlang(G_1)$ is the set of all languages in which the robot has enough battery to act.

Consider now the CFHG 
$G_2 = \la \{ a, c \}, \{ x_1, x_2\},\{ V_0, V_1\},V_0,P_2,\exists x_1 \forall x_2\ra$ where 
\begin{align*}
   & P_2:= 
   V_0 \rightarrow \tuplevar{c_{x_1},c_{x_2}} V_0 \tuplevar{a_{x_1},a_{x_2}} ~|~ \tuplevar{c_{x_1}, c_{x_2}} V_1 \tuplevar{a_{x_1}, \#_{x_2} } ~|~ \tuplevar{c_{x_1}, c_{x_2}} V_1 \tuplevar{a_{x_1}, a_{x_2} } 
   \\
    & V_1\rightarrow \tuplevar{c_{x_1}, \#_{x_2}} V_1 \tuplevar{a_{x_1}, \#_{x_2}} ~|~ \tuplevar{c_{x_1}, \#_{x_2}}~|~ \tuplevar{c_{x_1}, c_{x_2}}
\end{align*}
We now require that the robot only has one additional unit of charging (unlike in $G_1$). In addition, we require an upper bound (assigned to $x_1$) on the charging and action times. All other words in the language (assigned to $x_2$) correspond to shorter computations. 
\end{example}

We now study the nonemptiness and membership problems for CFHGs. 
When regarding a CFHG as a specification, these correspond to the model-checking and satisfiability problems. 



\begin{theorem} \label{thm:ExistsRankedEmptiness}\label{thm:existsemptiness}
The nonemptiness problem for $\exists^*$-CFHG is in \comp{P}. 
\end{theorem}

\begin{proof} 
According to the semantics of the $\exists$-requirement, an $\exists^*$-CFHG $G$ derives a language $\lang$ if $\hat G$ accepts a word assignment that corresponds to words in $\lang$. Therefore, it is easy to see that $G$ is nonempty iff $\hat G$ is nonempty. 
Since the nonemptiness of CFG is in \comp{P}~\cite{DBLP:books/daglib/Hopcroft}, we are done. 
\end{proof}


\begin{theorem} \label{thm:regmodelcheck}
The membership problem for a regular language in an $\exists^*$-CFHG is in \comp{EXPTIME}.
\end{theorem}

\begin{proof}
Let $A=\tuple{\Sigma,Q,Q_0,\delta,F}$ be an NFA and let $G$ be a CFHG with $\alpha = \exists x_1 \cdots\exists x_k$. 
In order to check whether $\reglang{A} \in \hlang(G)$, we need to check whether there exists a subset of $\reglang{A}$ of size $k$ or less, that can be accepted as a word assignment by $\hat{G}$. Since $\hat{G}$ derives words over $\Sigma\cup\{\#\}$, we first construct the NFA $A\uparrow{\#}$, that accepts all $\#$-paddings of words in $\lang(A)$. We can do so easily by adding a self-loop labeled $\#$ to every state in $A$. 
We then compute $(A\uparrow{\#})^{\otimes k}$ to allow different paddings for different words (an exponential construction), intersect the resulting automaton with $\hat{G}$ and test the intersection for nonemptiness.
Context-free languages are closed under intersection with regular languages via a polynomial construction. 
In addition, if the grammar is given in Chomsky Normal Form~\cite{DBLP:journals/iandc/Chomsky59a}, then checking the emptiness of the intersection is polynomial in the sizes of the grammar and automaton. As the conversation to Chomsky normal form is also polynomial, we get that the entire procedure is exponential, due to the size of $(A\uparrow{\#})^{\otimes k}$. 
\end{proof}

\begin{theorem}\label{thm:finitemembership}
The membership problem for a finite language in a CFHG is in \comp{EXPTIME}. 
\end{theorem}
\begin{proof} 
Let $\lang$ be a finite language and let $G$ be a CFHG with variables $\{x_1,\ldots x_k\}$. Since $\lang$ is finite, we can construct every assignment of words in $\lang$ to the variables in $G$, and check if it is accepted by $\hat{G}$.
Similarly to the proof of Theorem~\ref{thm:regmodelcheck}, to do so, we use an NFA ${A}_w$ whose language is the set $w\uparrow_\#$, for every $w\in \lang$.
For an assignment $v=[x_1\mapsto w_1]\dots[x_k\mapsto w_k]$, we construct
$
\bigotimes_{i=1}^k{A}_{w_i}$, 
and check the nonemptiness of its intersection with $\hat G$. As in Theorem~\ref{thm:regmodelcheck}, this procedure is exponential in the length of the words in $\lang$ and in $|G|$. 
Since we can finitely enumerate all assignments, we can check whether the quantification condition $\alpha$ of $G$ is satisfied. 
Enumerating all assignments amounts to traversing the decision tree dictated by $\alpha$, which is exponential in $|\alpha|$. 
Therefore, the entire procedure can be done in exponential time in $|G|$ and $\lang$. 
\end{proof}

\begin{theorem} \label{thm:undecforall}
The emptiness problem for 
$\forall^*$-CFHG and $\exists\forall$-CFHG is undecidable.
\end{theorem}

\begin{proof}
We show reductions from the Post correspondence problem (PCP).
A PCP instance is a set of pairs of the form $[a_1, b_1], \ldots ,[a_n, b_n]$ where $a_i, b_i\in \{a,b\}^*$. The problem is then to decide whether there exists a sequence of indices $i_1\cdots i_m, ~i_j\in[1,n]$, such that $a_{i_1}a_{i_2}\cdots a_{i_m} = b_{i_1}b_{i_2}\cdots b_{i_m}$. For example, consider the instance $\{[a, baa]_1, [ab, aa]_2, [bba, bb]_3 \}$. Then, a solution to the PCP is the sequence $3,2,3,1$ since $a_3a_2a_3a_1 = bba\cdot ab\cdot bba\cdot a$ and $b_3b_2b_3b_1 = bb\cdot aa\cdot bb\cdot baa$. 

Let $T = \{ [a_1, b_1], \ldots ,[a_n, b_n]\}$ be a PCP instance.  
Let $G = \tuple{\{a,b\},\{ x_1, x_2\}, \{V_0\}, V_0, P, \forall x_1 \forall x_2}$  be a $\forall^*$-CFHG defined as follows. 
For every pair $[a_i, b_i]\in T$ we define the words  $A_i, B_i\in (\Sigma\cup\{\#\})^*$
obtained from $a_i,b_i$ by padding the shorter of $a_i,b_i$ with $\#$-symbols so that $A_i,B_i$ are of equal length. 
We define $P$ as follows. 
\begin{align*}
  P:=  V_0 \rightarrow \tuplevar{{A_1}_{x_1}, {B_1}_{x_2}}V_0 ~|~ \cdots ~|~ \tuplevar{{A_n}_{x_1}, {B_n}_{x_2}}V_0 ~|~ \tuplevar{{A_1}_{x_1}, {B_1}_{x_2}} ~|~ \cdots ~|~ \tuplevar{{A_n}_{x_1}, {B_n}_{x_2}}
\end{align*}
For a language $\lang \in \hlang(G)$, it must hold that
$\wass{w}{[x_1\mapsto u][x_2\mapsto v]}\in\lang(\hat G)$ for every $u,v\in \lang$, due to the $\forall\forall$-condition. Let $u\in\lang$. Then, in particular, $\wass{w}{[x_1\mapsto u][x_2\mapsto u]}\in\reglang{\hat G}$. Notice that in this case, $u$ is a solution to $T$. In the other direction, a solution to $T$ induces a word $u=a_{i_1}a_{i_2}\cdots a_{i_m}$ such that $\{u\}\in\hlang(G)$.  
The same reduction holds also for the case of $\exists\forall$, since according to the $\forall$ requirement, one of the word assignments must assign the same word to both variables. 
\stam{
it has to hold that every word $w\in \lang$ can be assigned simultaneously to both $x_1$ and $x_2$, due to the $\forall x_1 \forall x_2$ condition. Such a simultaneous assignment $\wass{w}{[x_1 \mapsto w] [x_2\mapsto w]}$
is a solution of the PCP instance
since (1) the same word is generated both by the $A$ variables and by the $B$ variables; and (2) we only allow to derive tuples of the form $\tuple{A_i, B_i}$, and thus the sequence of applying the derivation rules corresponds to the set of indices that is the solution to the PCP instance. 
Therefore, every $\lang \in\hlang(G)$ forms a set of solutions to the PCP instance.

\textbf{$\forall\exists$-CFHG}. 
Let $G = \tuple{\Sigma,\{ x_1, x_2\}, \{V_0\}, V_0, P, \exists x_1 \forall x_2}$ be a CFHG where $P$ is the same as above. In a similar manner, if $L\in\hlang(G)$ then there must exist $w\in L$ such that $\wass{w}{[x_1 \mapsto w][x_2\mapsto w]}$ is a solution to the PCP instance, and vice versa. 
}
\end{proof}

Note that the proof of Theorem~\ref{thm:undecforall} 
compares between two words in order to simulate PCP. For a single $\forall$-quantifier, the nonemptiness problem is equivalent to that of CFG, and is therefore in \comp{P}.

The underlying CFG we use in the proof of Theorem~\ref{thm:undecforall} is linear, and so the result follows also to asynchronous NFH, that allow $\#$-symbols arbitrarily. This is in line with the results in~\cite{bcbfs21}, which shows that the model-checking problem for asynchronous hyperLTL is undecidable. 

\subsection{Synchronous Hypergrammars}\label{sec:sync}

As we show in Section~\ref{sec:CFGH}, the  asynchronicity of general CFHG leads to undecidability of most decision problems for them, already for simple quantification conditions. 
We now introduce {\em ranked CFHG}, a fragment of CFHG that ensures synchronous behavior. We then prove that ranked CFHG capture exactly the set of {\em synchronous hyperlanguages}. Intuitively, synchronous hyperlanguages are derived from grammars in which $\#$ only appears at the end of the word, similarly to NFH (we say that such a word assignment is {\em synchronous}). 
Since CFHG may use non-linear rules, in order to characterize the grammar rules that derive synchronous hyperlanguages, we need to reason about structural properties of the grammar. 
To this end, we define a {\em rank} for each variable $v$, which, intuitively, corresponds to word variables for which $v$ derives $\#$-symbols. 

\begin{remark}
Before we turn to the definition of ranks of variables and ranked grammars we note on the difference between a definition of grammars which their hyperlanguages are synchronous, as we do in the rest of this section; and the problem of, given some hypergrammar $G$, 
finding the hyperlanguage $\hlang{(G_s)} \subseteq \hlang{(G)}$ that corresponds to the synchronous sub-hyperlanguage of $G$. 
Assume that $G$ is over $\Sigma$ and has $k$ quantifiers. Then, the latter can be done by constructing an NFA $A_s$ over $(\Sigma \cup \{\#\})^k$ that accepts all words in which $\#$ appears only at the end of words. The intersection of $A_s$ and $G$ results in the grammar $G_s$, whose language is a subset of that of $G$. We approach a different problem, namely defining a fragment of grammars that accept exactly the class of synchronous hyperlanguages. 
\end{remark}

In order to define the ranks of variables, we use the \emph{rule graph} $\mathcal{G}$, defined as follows. The set of vertices of $\mathcal{G}$ is $V \cup W$, where
$W = \{ \gamma \in   (\hat\Sigma\cup V)^* ~|~ \exists v\in V. v\rightarrow\gamma\in P  \}$
is the set of sequences appearing on the right side of one of the grammar rules. 
The set of edges $E$ of $\mathcal{G}$ is $E = E_L \cup E_R$ where
\begin{align*}
    E_L = \{ \tuple{v, w} ~|~ v\rightarrow w \in P \}\cup \{ \tuple{w, v}\ ~|~ w = v\gamma \} \\
E_R = \{ \tuple{v, w} ~|~ v\rightarrow w \in P \}\cup \{ \tuple{w, v}\ ~|~ w = \gamma v\} 
\end{align*}
We partition $\mathcal{G}$ into maximal strongly connected components (MSCCs) with respect to each type of edges ($E_L$ and $E_R$), resulting in two directed a-cyclic graphs $\mathcal{G}_L$ and $\mathcal{G}_R$.
The vertices of $\mathcal{G}_d$ for $d\in\{L,R\}$ are the MSCCs according to $E_d$, 
and there is an edge $C_1^d\rightarrow C_2^d$ iff there exist $u,u'\in (V\cup W)$ such that $u\in C_1^d, u'\in C_2^d$ and $\langle u,u' \rangle\in E_d$. 
Note that every terminal word is a singleton MSCC in both graphs.

\begin{example}\label{ex:rank}
Figure~\ref{fig:MSCC} presents $\mathcal{G}_R$ for $G$
of the proof of Theorem~\ref{thm:undecforall}, and the PCP instance $\{[a, baa]_1, [ab, aa]_2,$ $[bba, bb]_3 \}$,
with the concrete derivation rules: 
\begin{align*}
  V_0 \rightarrow \tuplevar{{a\#\#}_{x_1}, {baa}_{x_2}}V_0 ~|~ \tuplevar{{ab}_{x_1}, {aa}_{x_2}}V_0 ~|~ 
  \tuplevar{{baa}_{x_1}, {bb\#}_{x_2}}V_0 ~|~ \\ \tuplevar{{a\#\#}_{x_1}, {baa}_{x_2}}~|~ \tuplevar{{ab}_{x_1}, {aa}_{x_2}} ~|~ \tuplevar{{baa}_{x_1}, {bb\#}_{x_2}}
\end{align*}
\end{example}


\begin{figure}[t]
\centering
\scalebox{.52}{
        \includegraphics[scale=0.9]{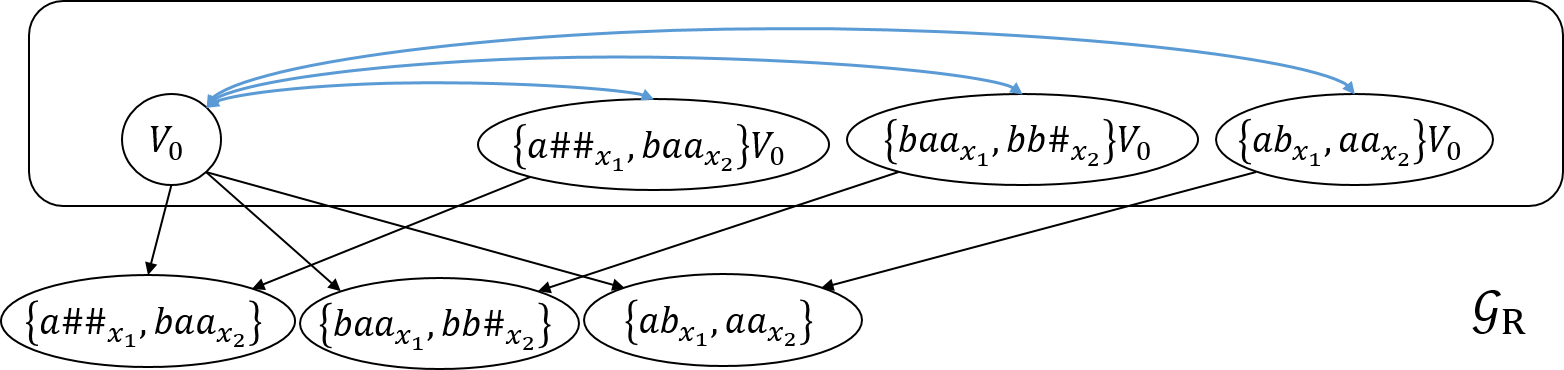}}
   \caption{The MSSC graph $\mathcal{G}_R$ for the grammar and PCP instance of Example~\ref{ex:rank} and the Proof of Theorem~\ref{thm:undecforall}. Blue edges are bidirectional, and the rectangle
   represents an MSCC. }
    \label{fig:MSCC}
\end{figure}

We now define the {\em left ranks} and {\em right ranks} of synchronous words, variables and sequences. 

\begin{enumerate}
    \item \textbf{Ranks of terminal synchronous words}. The rank of a letter $\hat\sigma =
    \{\sigma_{1_{x_1}}, \ldots \sigma_{n_{x_n}}\}\in \hat\Sigma$
    is $t(\hat\sigma) = \{x_i~|~\sigma_{i_{x_i}}= \#\}$. 
    The left rank of $\hat{w}$ is $\Le{\hat w} = t(\hat\sigma_1)$, and its right rank is $\Ri{\hat w} = t(\hat\sigma_n)$, where $\hat\sigma_1$ and $\hat\sigma_n$ are the first and last letters of $\hat w$, respectively.
\item \textbf{Inductive definition for variables and sequences}. 
Let $d\in \{L,R\}$, and let $C_1^d\rightarrow C_2^d$ 
in $\mathcal{G}_d$ such that $\mathbf{d}(u')$ 
is defined for every $u'\in C_2^d$ and $\mathbf{d}\in\{\mathbf{L}, \mathbf{R} \}$. 
Let $\gamma\in(\hat{\Sigma}\cup V)^*$, $\sigma\in\hat{\Sigma}$, and $v\in V$. 
\begin{itemize}
       \item For $u\in C_1^d\in \mathcal{G}_d$ such that $u = \sigma \gamma$
      we define $\Le{u} = l(u) = \Le{\sigma}$.
     \item For $u\in C_1^d\in \mathcal{G}_d$ such that $u = \gamma \sigma $
     we define
    $\Ri{u} = r(u) = \Ri{\sigma}$.
        \item For $u\in C_1^L \in \mathcal{G}_L$ such that $u = v\gamma$
    we define $l(u) =\bigcup_{C_1^L\rightarrow C_2^L} \bigcap_{u'\in C_2^L} \Le{u'}$.
    \item For $u\in C_1^R\in\mathcal{G}_R$ such that $u = \gamma v$
    we define
$r(u) = \bigcup_{C_1^R\rightarrow C_2^R} \bigcup_{u'\in C_2 ^R} \Ri{u'}$. 
\end{itemize}
Now, for each $u = v\gamma \in C_1^L$  
we define $\Le{u} = \bigcap_{u'\in C_1^L} l(u')$, 
and for each 
 $u = \gamma v \in C_1^R$ we define  
$\Ri{u} = \bigcup_{u'\in C_1^R} r(u')$. 
\end{enumerate}

Note that this process is guaranteed to terminate, since we traverse both graphs in reverse topological order. 
Therefore, at the end of the process, $\Le{u}$ and
$\Ri{u}$ are defined for every $u \in V\cup W$.  


We define \emph{ranked CFGs} to be CFGs in which for every rule $v\rightarrow \gamma_1 \cdots \gamma_n$ for $\gamma_i\in (\hat\Sigma \cup V)$, it holds that $\Ri{\gamma_i}\subseteq \Le{\gamma_{i+1}}$. Intuitively, this means that $\gamma_i$ may not produce $\#$ to its right, if $\gamma_{i+1}$ can produce $\sigma\neq \#$ to its left, leading to unsynchronous derivation.  
A CFHG $G$ is \emph{ranked} if $\hat G$ is ranked.

\begin{example}
Consider $G$ of Example~\ref{ex:rank} and  $\mathcal{G}_R$ of Figure~\ref{fig:MSCC}. The graph $\mathcal{G}_L$ is similar to $\mathcal{G}_R$, with no edges back to $V_0$, (and thus without the rectangle MSCC). We compute some of the ranks for $G$:
\begin{gather*}
    \Le{\{a\#\#_{x_1}, baa_{x_2} \}} = \Le{\{baa_{x_1}, bb\#_{x_2} \}} = \Le{\{ab_{x_1}, aa_{x_2} \}} = \emptyset ~~~~~  \Le{V_0} = \emptyset \\
    \Ri{\{a\#\#_{x_1}, baa_{x_2} \}} = \{ x_1 \}~~~~\Ri{\{baa_{x_1}, bb\#_{x_2} \}} = \{ x_2 \}~~~~ \Ri{\{ab_{x_1}, aa_{x_2} \}} = \emptyset  ~~~~ \Ri{V_0} = \{x_1,x_2 \}
\end{gather*}
 $G$ is not ranked, since for the rule $V_0 \rightarrow \{a\#\#_{x_1}, baa_{x_2} \} V_0$,
it holds that $\Ri{\{a\#\#_{x_1}, baa_{x_2} \} }\not\subseteq \Le{V_0}$.
\end{example}

\begin{example}
The following CFHG $G_r = \langle \{a,b \}, \{x_1, x_2 \}, \{V_0, V_1\}, V_0, P, \forall x_1\exists x_2\rangle$ is ranked, where $P$ is: 
\begin{align*}
    & P:= 
    V_0 \rightarrow V_1 V_2
    \\ & V_1 \rightarrow \{a_{x_1}, a_{x_2} \} V_1 \{ b_{x_1}, b_{x_2}\} ~|~ \{ab_{x_1}, ab_{x_2} \} 
    \\ 
    & V_2 \rightarrow  V_2 \{ \#_{x_1}, b_{x_2}\}~ | ~\{ \#_{x_1}, b_{x_2}\} 
\end{align*}
$G_r$ accepts all languages in which for every word of the type $a^n b^n$ there exits a word with more $b$'s, that is, there exists $a^n b^m$ for $m>n$.

The ranks of $G_r$, as shown below, demonstrate that $G_r$ is indeed ranked.  
\begin{gather*}
     \Le{\{ \#_{x_1}, b_{x_2}\}} =  \Ri{\{ \#_{x_1}, b_{x_2}\}} = \{ x_1\} ~~~~~ \Le{\{ab_{x_1}, ab_{x_2} \}  } =  \Ri{\{ab_{x_1}, ab_{x_2} \} } =  \emptyset
     \\ \Le{\{a_{x_1}, a_{x_2} \} V_1 \{ b_{x_1}, b_{x_2}\}} = \Ri{\{a_{x_1}, a_{x_2} \} V_1 \{ b_{x_1}, b_{x_2}\}} = \Ri{V_1} = \Le{V_1} = \emptyset 
     \\ 
     \Ri{V_2 \{\#_{x_1}, b_{x_2} \}} = \Le{V_2 \{\#_{x_1}, b_{x_2} \}} = \Ri{V_2} = \Le{V_2} = \{ x_1 \}  
      \\
     \Ri{V_0} = \emptyset ~~~ \Le{V_0} = \{x_1\} 
\end{gather*}
\end{example}

\begin{definition}
$\hlang$ is a {\em synchronous context-free hyperlanguage} if 
there exists a CFHG $G$ for $\hlang$ in which $\hat G$ only derives synchronous word assignments.
\end{definition}

\begin{theorem}\label{prop:rankedgrammars}
A hyperlanguage $\hlang$ is derived by a ranked CFHG iff $\hlang$ is synchronous context-free.  
\end{theorem}

In order to prove Theorem~\ref{prop:rankedgrammars}, we use the following claims. 

\begin{claim} \label{claim:ranked}
Let $G = \tuple{\Sigma, X, V, V_0, P, \alpha}$
be a ranked CFHG. Then, for every word $\gamma = \gamma_1 \cdots \gamma_n\in (\hat{\Sigma}\cup V)^*$, if there exists $v\in V$ such that $v \Rightarrow^* \gamma$, then $\Ri{\gamma_i}\subseteq \Le{\gamma_{i+1}}$ for all $i\in[1, n-1]$.
\end{claim}

\begin{claim}\label{claim:derive}
Let $G = \tuple{\Sigma, X, V, V_0, P, \alpha}$
be a (possibly not ranked) CFHG with $|X| = k$, and let $v\in V$.
\begin{enumerate}
    \item For every $j\in [1,k]\setminus \Le{v}$ there exists $w\in\hat{\Sigma}^*$ such that $v\Rightarrow^* w$ and $j\notin \Le{w}$. 
    \item For every $j\in\Ri{v}$ there exists $w\in\hat{\Sigma}^*$ such that $v\Rightarrow^* w$ and $j\in \Ri{w}$. 
\end{enumerate}
\end{claim}

\begin{proof}[Proof of Theorem~\ref{prop:rankedgrammars}]
Let $\hlang$ be a context-free language that is accepted by a ranked grammar $G$. 
According to Claim~\ref{claim:ranked}, 
for every word $w = w_1 \cdots w_n\in \hat{\Sigma}^*$ such that $V_0 \Rightarrow^* w$, it holds that $\Ri{w_i}\subseteq \Le{w_{i+1)}}$ for $i\in[1, n-1]$. That is, $\#$ is allowed to only appear at the end of words, and so $\hat G$ only derives synchronous word assignments. 


For the other direction, let $\hlang$ be a synchronous context-free hyperlanguage, and let $G$ be a CFHG for $\hlang$ that only derives synchronous word assignments. 
Assume by way of contradiction that $G$ is not ranked. Then, there exists some rule 
$v\rightarrow \gamma_1 \cdots \gamma_n \in P$ 
where $\gamma_i \in (\hat{\Sigma} \cup V)$
such that $\Ri{\gamma_i}\not\subseteq \Le{\gamma_{i+1}}$ for some $i\in[1,n]$. 
Recall that we assume that all rules are reachable and that every variable can derive a terminal word. 
Consider a derivation sequence $V_0 \Rightarrow^* \beta v \beta' \Rightarrow \beta\gamma_1 \cdots \gamma_n \beta'$. 
Then, there exist $w,w_i,w_{i+1}, w'\in \hat{\Sigma}^*$ such that $\gamma_i \Rightarrow^* w_i$, $\gamma_{i+1} \Rightarrow^* w_{i+1}$ and $V_0 \Rightarrow^* w w_i w_{i+1} w'$; and due to claim~\ref{claim:derive}, for some $j\in \Ri{\gamma_i} \setminus \Le{\gamma_{i+1}}$, it holds that $j\in \Ri{w_i}\setminus \Le{w_{i+1}}$. Hence, $w_i$ ends with $\#$ in some location which is followed by a letter in $w_{i+1}$, and so the word assignment $w w_i w_{i+1} w'$ is not synchronous, a contradiction.
\end{proof}

We therefore term ranked grammars {\em syncCFHG}. 
Given a CFHG $G$, deciding whether it is ranked amounts to constructing the graph $\mathcal{G}$ and traversing the topological sorting of its MSCC graph in reverse order in order to compute all ranks, and finally checking that all grammar rules of $G$ comply to the rank rules. All these steps can be computed in polynomial time. We now show that syncCFHG is more decidabile than CFHG.

\stam{
\begin{proposition}
The problem of deciding whether a CFHG $G$ defines a synchronous hyperlanguage is in \comp{P}.
\end{proposition}

\begin{proof}
The construction of the graph $\mathcal{G}$ and the partition into MSCCs is polynomial in the number of variables and in the number of rules in $P$.  Then, computing the ranks of variables is done by (1) computing $\bigcup_{C} \bigcup_{v'\in C } r(v')$ and $\bigcup_{C} \bigcap_{v'\in C} l(v')$ for each component that was already traversed; and (2) traverse each vertex once, and for this vertex set $r(v)$ and $l(v)$. 
Last, we need to check that all rules $v\rightarrow \gamma_1 \cdots \gamma_n$  satisfy the requirement $r(\gamma_i)\subseteq l(\gamma_{i+1})$. This check is polynomial in the number of rules and the size of the longest rule.  
\end{proof}
}

\begin{theorem}\label{thm:forallsync}
The nonemptiness problem for $\forall^*$-syncCFHG and $\exists\forall^*$-syncCFHG is in \comp{P}.
\end{theorem}

\begin{proof}
Let $G$  be a syncCFHG. 
Since universal quantification is closed under subsets, it holds that if $\lang \in \hlang(G)$, then $\lang' \in \hlang(G)$ for every $\lang'\subseteq \lang$.
Therefore, it suffices to check whether there exists a singleton $\lang$ 
such that $\lang\in\hlang(G)$. 
Therefore, we consider only word assignments of the form ${\bi w} = \wass{w}{[x_1\mapsto w]\cdots[x_k\mapsto w]}$ for some $w\in (\Sigma\cup\{\#\})^*$. 
Notice that $\bi w$ has a single representation, since $\#$ may not appear arbitrarily. 
We construct a syncCFHG $G'$ by restricting $\hat G$ to the alphabet $\bigcup_{\sigma\in \Sigma} \{\sigma\}^X$, that is, all variables are assigned the same letter. 
All rules over other alphabet letters are eliminated. Since elimination of rules cannot induce asynchronization, $G'$ is synchronous.

Now, for a singleton language $\{w\}$, we have $\{w\}\in\hlang(G)$ iff $\{w\}\in\hlang(G')$. Therefore, it suffices to check the nonemptiness of $\hlang(G')$, which amounts to checking the nonemptiness of $\hat G'$. 

The proof holds also for the case of $\exists\forall^*$-syncCFHG. Indeed, an $\exists\forall^*$-syncCFHG $G$ is nonempty iff it derives a singleton hyperlanguage. This, since in a language derived by $G$, a word $w$  that is assigned to the variable under $\exists$ must also be assigned to all variables under $\forall$ in one of the word assignments derived by $\hat G$,  which in turn fulfills the requirements for deriving $\{w\}$. Since $G$ is synchronous, it suffices to restrict the alphabet to homogeneous letters and check for nonemptiness, as with $\forall^*$.  
\stam{

Now, if $\wass{w}{[x_1\mapsto w]\cdots[x_k\mapsto w]} \in\reglang{\hat{G'}}$ then it is also in $\reglang{\hat{G}}$ and thus $S= \{ w\downarrow_{\#}\}\in\hlang{(G)}$. This can be proved inductively using the semantics of CFHG, similar to the proof of Theorem~\ref{thm:existsemptiness}. 

On the other hand, we show that all word assignments in $\hat{G}$ of the form $\wass{w}{u} = \wass{w}{[x_1\mapsto w]\cdots[x_k\mapsto w]}$ are also accepted by $\hat{G'}$. Assume by contradiction that there is such $\wass{w}{u}\in\hat{G}\setminus\hat{G'}$. From the construction of $G'$, it holds that when deriving $\wass{w}{u}$, we have to apply a rule that is eliminated from $G'$. However this would imply that different letters were assigned to two occurrences of $w_i$ (the $i$'th letter of $w$) for some $i\leq|w|$, and since $G$ is synchronous, this implies deriving two different words, which is a contradiction. 
Therefore, all word assignments of the from $\wass{w}{[x_1\mapsto w]\cdots[x_k\mapsto w]}$ that are accepted by $\hat{G}$, are also accepted by $\hat{G'}$, and every language $S = \{ w\}\in\hlang(G)$ of size one corresponds to the word assignment $\wass{w}{[x_1\mapsto w]\cdots[x_k\mapsto w]}\in\reglang{\hat{G'}}$. 

We have that $\reglang{\hat{G'}} \neq \emptyset $ iff $\hlang{(G)}$ contains a language of size one, iff $\hlang{(G)}\neq \emptyset$. Therefore we can reduce the nonemptiness of $G$ to the nonemptiness of $\hat{G'}$, which can be done in polynomial time in $\hat{G'}$ and thus in $G$.
}
\end{proof}


In the {\em regular membership problem}, we ask whether a regular language $\lang$ can be derived by a synchCFHG $G$.
This problem is decidable for NFH~\cite{bs21}.
For $\lang=\Sigma^*$ and a $\forall$-CFHG $G$, the question amounts to checking the universality of $\hat G$, which is  undecidable~\cite{DBLP:journals/jcss/BakerB74}. Therefore, we have the following. 

\begin{theorem}\label{thm:forallsyncundec}
The regular membership problem for $\forall^*$-syncCFHG grammars is undecidable.  
\end{theorem}

\stam{
\begin{proof}
We show a reduction to the universality problem of CFG, which is undecidable~\cite{DBLP:journals/jcss/BakerB74}. 
Consider the language $\Sigma^*$ and let $G$ be a $\forall$-syncCFHG. In order to check if $\Sigma^* \in\hlang(G)$, we need to verify that every $w\in\Sigma^*$ is accepted by the underlying grammar $\hat{G}$. That is, $\Sigma^*\in\hlang(G)$ iff $\hat{G}$ is universal. 
\end{proof}
}
\begin{remark} \label{rem:finitemem}
 Membership of a finite language $\lang$ in a CFHG with any quantification condition is decidable already for general CFHG (Theorem~\ref{thm:finitemembership}), with exponential complexity. 
 For syncCFHG, we can reduce the complexity by checking membership of word assignments instead. This, since we only need to consider synchronous words, which have a single representation. 
 Since checking membership is polynomial in the size of the word (for a grammar of fixed size)~\cite{DBLP:books/daglib/Hopcroft}, every such test is then polynomial. Since we may still need to traverse all possible word assignment, the complexity is exponential in the length of the quantification condition, but is polynomial in $|G|$ and the size of the words in $\lang$. 
 \end{remark}

  \begin{remark}\label{rem:regmemsync}
  For regular languages and $\exists^*$-CFHG, synchronization does not avoid the composition of automata, and we use a construction similar to the one of Theorem~\ref{thm:regmodelcheck}.
\end{remark}

We now show that synchronicity does not suffice for deciding nonemptiness of $\exists^*\forall^*$-syncCFHG. 

\begin{theorem}\label{thm:emptinessexistsforall}
The nonemptiness problem for $\exists^*\forall^*$-syncCFHG is undecidable.
\end{theorem}

\begin{proof}
We reduce from PCP. 
Let $T = \{ [a_1, b_1], \ldots ,[a_n, b_n]\}$ be a PCP instance over $\{ a,b\}$,
and let $$G = \tuple{\{a,b, c\}\cup[1,n],\{ x_1, x_2, x_3\}, \{V_0, V_1, V_2\}, V_0, P, \exists x_1\exists x_2 \forall x_3}$$ be a CFHG where $P$ is defined as follows. 
\begin{align*}
  &P:=  V_0 \rightarrow V_1 ~|~ V_2 \\ 
  &V_1 \rightarrow \tuplevar{{a_i}_{x_1}, {c^{|a_i|}}_{x_2} ,{a_i}_{x_3}}V_1 \tuplevar{i_{x_1}, c_{x_2}, i_{x_3}} 
  ~|~  \tuplevar{{a_i}_{x_1}, {c^{|a_i|}}_{x_2} ,{a_i}_{x_3}}\tuplevar{i_{x_1}, c_{x_2}, i_{x_3}} ~~\forall i\in[1,n] \\
   &V_2 \rightarrow \tuplevar{{b_i}_{x_1}, {c^{|b_i|}}_{x_2} ,{c^{|b_i|}}_{x_3}}V_2 \tuplevar{i_{x_1}, c_{x_2}, c_{x_3}} ~|~    \tuplevar{{b_i}_{x_1}, {c^{|b_i|}}_{x_2} ,{c^{|b_i|}}_{x_3}} \tuplevar{i_{x_1}, c_{x_2}, c_{x_3}} ~~\forall i\in[1,n] 
\end{align*}
Since none of the rules include the $\#$-symbol, $G$ is indeed a syncCFHG. 
Now, if there exists $\lang\in \hlang(G)$, then there exist $w \in \{a,b\}^*\cdot[1,n]^*$ and $w_c\in \{c\}^*$ both in $\lang$, such that for every $w'\in \lang$,
we have $V_0\Rightarrow^* \wass{w}{[x_1\mapsto w][x_1\mapsto w_c][x_3\mapsto w']}$. 
In particular, for $w'=w$, we have $V_0\Rightarrow^* \wass{w}{[x_1\mapsto w][x_1\mapsto w_c][x_3\mapsto w]}$. 
Since $V_2$ only derives words of the form $c^k$ in $x_3$, 
the derivation of $\wass{w}{[x_1\mapsto w][x_2\mapsto w_c][x_3\mapsto w]}$ is of the form $V_0 \Rightarrow V_1 \Rightarrow^* \wass{w}{[x_1\mapsto w][x_2\mapsto w_c][x_3\mapsto w]}$. In addition, all words in $\{ c\}^*$ can only be assigned to $x_3$ if derived from $V_2$, thus we have $V_0\Rightarrow V_2 \Rightarrow^*\wass{w}{[x_1\mapsto w][x_2\mapsto w_c][x_3\mapsto w_c]}$. 
Denote $w = w_1 w_2$ where $w_1\in\{a,b\}^*, w_2\in[1,n]$. Then,~$w$ encodes a solution to $T$, where $w_1$ is the string obtained from $a_i$ (and $b_i$), and $w_2$ is the sequence of~indices. 

For the other direction, a solution to $T$ encoded by a string $w_1$ and sequence of indices $w_2$ corresponds to the language $\{ w_1w_2, c^{|w_1w_2|}\}$ that is accepted by $G$. 
\stam{
=====

Since $V_2$ only derives words of the form $c^k$ in the $\forall$ assignment, all words $w'\in\{a,b\}^*$ assigned to $x_3$ have to be derived using $V_1$. 
In addition, all words of the form $w' = c^k$ can only be assigned to $x_3$ when derived from $V_2$. That is, there must exist $w\in\{a,b\}^*$ that is always assigned to $x_1$ such that $V_2 \Rightarrow^* \wass{w}{[x_1\mapsto w][x_2\mapsto c^{|w|}][x_3\mapsto c^{|w|}]}$
and in addition, there exists $w\in\{a,b\}^*$ such that $V_1 \Rightarrow^* \tuple{w, c^{|w|}, w'}$. Since $V_1$ only derives the same letters for $x_1$ and $x_3$, it holds that $w = w'$, and thus the derivation sequence is a solution to $T$. For the other direction, if there is a solution to $T$, then there exist such $w$ that can be derived using the two variables, and thus the language of $G$ is not empty. }
\end{proof}

The nonemptiness problem for $\forall^*\exists^*$-NFH is undecidable~\cite{bs21}. Therefore, this is also the case for syncCFHG, and for general CFHG. 
\stam{
\subsection{Closure Properties}

As CFG are not closed under intersection, this is also the case for CFHGs. We now show that CFHGs are closed under union.
\hadar{verify the implication from CFG to CFHG}

\begin{theorem}
CFHG are closed under union. 
\end{theorem}

\begin{proof}
We can follow the union construction for CFG, talk about sets of languages. 
\end{proof}

\begin{theorem}
synchronous grammars are closed under union. 

\end{theorem}

\begin{proof}
need to take care of $\#$ and add a lot of variables
\end{proof}

\hadar{can we say something about intersection of CFHG with hyperautomata? similar to theorem~\ref{thm:existsemptiness}}
}

\section{Discussion and Future Work}

We have studied the realizability problem for regular hyperlanguages, focusing on the case of singleton hyperlanguages. We have shown that simple quantification conditions cannot realize this case. We have defined ordered and partially-ordered languages, for which we can construct hyperautomata that enumerate the language by order. We have shown that all regular languages are partially ordered. Since regular hyperlanguages are closed under union~\cite{bs21}, the result extends to a finite hyperlanguage containing regular languages. Naturally, there are richer cases one can consider. For an infinite hyperlanguage $\hl$, some characterization on the elements of $\hl$ would need to be defined in order to explore its realizability. We plan on pursuing this direction as future work. Another related direction is finding techniques for proving unrealizability for certain quantification conditions, for various types of hyperlanguages. 

In the second part of the paper we have studied the natural extension of context-free grammars to handle context-free hyperlanguages. Here, we have shown that beyond the inherent undecidability of some decision problems for hypergrammars, some undecidability properties stem from the asynchronous nature of these hypergrammars. We have then defined a synchronous fragment of context-free hyperlanguages, and defined a fragment of context-free grammars which exactly captures this fragment. The result retains some of the decidability properties of context-free grammars. As a future direction, we plan to study the realizability problem for CFHG and syncCFHG. Due to the limited closure properties of CFG, this is expected to be more challenging than for NFH. Another possible future direction is studying the entire Chomsky hierarchy for hyperlanguages, and finding fragments of the extensions to hyperlanguages that conserve the properties of these models for standard languages. 

\paragraph{Acknowledgements.}
We thank the anonymous reviewer for suggesting the elegant construction mentioned in Remark~\ref{rem:automatic}.

\bibliographystyle{eptcs}
\bibliography{biblio}

\begin{thebibliography}{10}
\providecommand{\bibitemdeclare}[2]{}
\providecommand{\surnamestart}{}
\providecommand{\surnameend}{}
\providecommand{\urlprefix}{Available at }
\providecommand{\url}[1]{\texttt{#1}}
\providecommand{\href}[2]{\texttt{#2}}
\providecommand{\urlalt}[2]{\href{#1}{#2}}
\providecommand{\doi}[1]{doi:\urlalt{http://dx.doi.org/#1}{#1}}
\providecommand{\eprint}[1]{arXiv:\urlalt{https://arxiv.org/abs/#1}{#1}}
\providecommand{\bibinfo}[2]{#2}

\bibitemdeclare{article}{as85}
\bibitem{as85}
\bibinfo{author}{B.~\surnamestart Alpern\surnameend} \& \bibinfo{author}{F.B.
  \surnamestart Schneider\surnameend} (\bibinfo{year}{1985}):
  \emph{\bibinfo{title}{Defining Liveness}}.
\newblock {\sl \bibinfo{journal}{Information Processing Letters}}, pp.
  \bibinfo{pages}{181--185}, \doi{10.1016/0020-0190(85)90056-0}.

\bibitemdeclare{article}{DBLP:journals/jcss/BakerB74}
\bibitem{DBLP:journals/jcss/BakerB74}
\bibinfo{author}{Brenda~S. \surnamestart Baker\surnameend} \&
  \bibinfo{author}{Ronald~V. \surnamestart Book\surnameend}
  (\bibinfo{year}{1974}): \emph{\bibinfo{title}{Reversal-Bounded Multipushdown
  Machines}}.
\newblock {\sl \bibinfo{journal}{J. Comput. Syst. Sci.}}
  \bibinfo{volume}{8}(\bibinfo{number}{3}), pp. \bibinfo{pages}{315--332},
  \doi{10.1016/S0022-0000(74)80027-9}.

\bibitemdeclare{inproceedings}{bcbfs21}
\bibitem{bcbfs21}
\bibinfo{author}{Jan \surnamestart Baumeister\surnameend},
  \bibinfo{author}{Norine \surnamestart Coenen\surnameend},
  \bibinfo{author}{Borzoo \surnamestart Bonakdarpour\surnameend},
  \bibinfo{author}{Bernd \surnamestart Finkbeiner\surnameend} \&
  \bibinfo{author}{C{\'{e}}sar \surnamestart S{\'{a}}nchez\surnameend}
  (\bibinfo{year}{2021}): \emph{\bibinfo{title}{A Temporal Logic for
  Asynchronous Hyperproperties}}.
\newblock In \bibinfo{editor}{Alexandra \surnamestart Silva\surnameend} \&
  \bibinfo{editor}{K.~Rustan~M. \surnamestart Leino\surnameend}, editors: {\sl
  \bibinfo{booktitle}{Computer Aided Verification - 33rd International
  Conference, {CAV} 2021, Virtual Event, July 20-23, 2021, Proceedings, Part
  {I}}}, {\sl \bibinfo{series}{Lecture Notes in Computer Science}}
  \bibinfo{volume}{12759}, \bibinfo{publisher}{Springer}, pp.
  \bibinfo{pages}{694--717}, \doi{10.1007/978-3-030-81685-8\_33}.

\bibitemdeclare{inproceedings}{bss18}
\bibitem{bss18}
\bibinfo{author}{Borzoo \surnamestart Bonakdarpour\surnameend},
  \bibinfo{author}{C{\'{e}}sar \surnamestart S{\'{a}}nchez\surnameend} \&
  \bibinfo{author}{Gerardo \surnamestart Schneider\surnameend}
  (\bibinfo{year}{2018}): \emph{\bibinfo{title}{Monitoring Hyperproperties by
  Combining Static Analysis and Runtime Verification}}.
\newblock In \bibinfo{editor}{Tiziana \surnamestart Margaria\surnameend} \&
  \bibinfo{editor}{Bernhard \surnamestart Steffen\surnameend}, editors: {\sl
  \bibinfo{booktitle}{Leveraging Applications of Formal Methods, Verification
  and Validation. Verification - 8th International Symposium, ISoLA 2018,
  Limassol, Cyprus, November 5-9, 2018, Proceedings, Part {II}}}, {\sl
  \bibinfo{series}{Lecture Notes in Computer Science}} \bibinfo{volume}{11245},
  \bibinfo{publisher}{Springer}, pp. \bibinfo{pages}{8--27},
  \doi{10.1007/978-3-030-03421-4\_2}.

\bibitemdeclare{inproceedings}{bs21}
\bibitem{bs21}
\bibinfo{author}{Borzoo \surnamestart Bonakdarpour\surnameend} \&
  \bibinfo{author}{Sarai \surnamestart Sheinvald\surnameend}
  (\bibinfo{year}{2021}): \emph{\bibinfo{title}{Finite-Word Hyperlanguages}}.
\newblock In \bibinfo{editor}{Alberto \surnamestart Leporati\surnameend},
  \bibinfo{editor}{Carlos \surnamestart Mart{\'{\i}}n{-}Vide\surnameend},
  \bibinfo{editor}{Dana \surnamestart Shapira\surnameend} \&
  \bibinfo{editor}{Claudio \surnamestart Zandron\surnameend}, editors: {\sl
  \bibinfo{booktitle}{Language and Automata Theory and Applications - 15th
  International Conference, {LATA} 2021, Milan, Italy, March 1-5, 2021,
  Proceedings}}, {\sl \bibinfo{series}{Lecture Notes in Computer Science}}
  \bibinfo{volume}{12638}, \bibinfo{publisher}{Springer}, pp.
  \bibinfo{pages}{173--186}, \doi{10.1007/978-3-030-68195-1\_17}.

\bibitemdeclare{inproceedings}{DBLP:conf/concur/BouajjaniER94}
\bibitem{DBLP:conf/concur/BouajjaniER94}
\bibinfo{author}{Ahmed \surnamestart Bouajjani\surnameend},
  \bibinfo{author}{Rachid \surnamestart Echahed\surnameend} \&
  \bibinfo{author}{Riadh \surnamestart Robbana\surnameend}
  (\bibinfo{year}{1994}): \emph{\bibinfo{title}{Verification of Nonregular
  Temporal Properties for Context-Free Processes}}.
\newblock In \bibinfo{editor}{Bengt \surnamestart Jonsson\surnameend} \&
  \bibinfo{editor}{Joachim \surnamestart Parrow\surnameend}, editors: {\sl
  \bibinfo{booktitle}{{CONCUR} '94, Concurrency Theory, 5th International
  Conference, Uppsala, Sweden, August 22-25, 1994, Proceedings}}, {\sl
  \bibinfo{series}{Lecture Notes in Computer Science}} \bibinfo{volume}{836},
  \bibinfo{publisher}{Springer}, pp. \bibinfo{pages}{81--97},
  \doi{10.1007/978-3-540-48654-1\_8}.

\bibitemdeclare{article}{DBLP:journals/iandc/Chomsky59a}
\bibitem{DBLP:journals/iandc/Chomsky59a}
\bibinfo{author}{Noam \surnamestart Chomsky\surnameend} (\bibinfo{year}{1959}):
  \emph{\bibinfo{title}{On Certain Formal Properties of Grammars}}.
\newblock {\sl \bibinfo{journal}{Inf. Control.}}
  \bibinfo{volume}{2}(\bibinfo{number}{2}), pp. \bibinfo{pages}{137--167},
  \doi{10.1016/S0019-9958(59)90362-6}.

\bibitemdeclare{book}{DBLP:books/daglib/0007403-2}
\bibitem{DBLP:books/daglib/0007403-2}
\bibinfo{author}{Edmund~M. \surnamestart Clarke\surnameend},
  \bibinfo{author}{Orna \surnamestart Grumberg\surnameend},
  \bibinfo{author}{Daniel \surnamestart Kroening\surnameend},
  \bibinfo{author}{Doron~A. \surnamestart Peled\surnameend} \&
  \bibinfo{author}{Helmut \surnamestart Veith\surnameend}
  (\bibinfo{year}{2018}): \emph{\bibinfo{title}{Model checking, 2nd Edition}}.
\newblock \bibinfo{publisher}{{MIT} Press}.
\newblock
  \urlprefix\url{https://mitpress.mit.edu/books/model-checking-second-edition}.

\bibitemdeclare{inproceedings}{DBLP:conf/post/ClarksonFKMRS14}
\bibitem{DBLP:conf/post/ClarksonFKMRS14}
\bibinfo{author}{Michael~R. \surnamestart Clarkson\surnameend},
  \bibinfo{author}{Bernd \surnamestart Finkbeiner\surnameend},
  \bibinfo{author}{Masoud \surnamestart Koleini\surnameend},
  \bibinfo{author}{Kristopher~K. \surnamestart Micinski\surnameend},
  \bibinfo{author}{Markus~N. \surnamestart Rabe\surnameend} \&
  \bibinfo{author}{C{\'{e}}sar \surnamestart S{\'{a}}nchez\surnameend}
  (\bibinfo{year}{2014}): \emph{\bibinfo{title}{Temporal Logics for
  Hyperproperties}}.
\newblock In \bibinfo{editor}{Mart{\'{\i}}n \surnamestart Abadi\surnameend} \&
  \bibinfo{editor}{Steve \surnamestart Kremer\surnameend}, editors: {\sl
  \bibinfo{booktitle}{Principles of Security and Trust - Third International
  Conference, {POST} 2014, Held as Part of the European Joint Conferences on
  Theory and Practice of Software, {ETAPS} 2014, Grenoble, France, April 5-13,
  2014, Proceedings}}, {\sl \bibinfo{series}{Lecture Notes in Computer
  Science}} \bibinfo{volume}{8414}, \bibinfo{publisher}{Springer}, pp.
  \bibinfo{pages}{265--284}, \doi{10.1007/978-3-642-54792-8\_15}.

\bibitemdeclare{article}{cs10}
\bibitem{cs10}
\bibinfo{author}{Michael~R. \surnamestart Clarkson\surnameend} \&
  \bibinfo{author}{Fred~B. \surnamestart Schneider\surnameend}
  (\bibinfo{year}{2010}): \emph{\bibinfo{title}{Hyperproperties}}.
\newblock {\sl \bibinfo{journal}{J. Comput. Secur.}}
  \bibinfo{volume}{18}(\bibinfo{number}{6}), pp. \bibinfo{pages}{1157--1210},
  \doi{10.3233/JCS-2009-0393}.

\bibitemdeclare{inproceedings}{DBLP:conf/lics/CoenenFHH19}
\bibitem{DBLP:conf/lics/CoenenFHH19}
\bibinfo{author}{Norine \surnamestart Coenen\surnameend},
  \bibinfo{author}{Bernd \surnamestart Finkbeiner\surnameend},
  \bibinfo{author}{Christopher \surnamestart Hahn\surnameend} \&
  \bibinfo{author}{Jana \surnamestart Hofmann\surnameend}
  (\bibinfo{year}{2019}): \emph{\bibinfo{title}{The Hierarchy of Hyperlogics}}.
\newblock In: {\sl \bibinfo{booktitle}{34th Annual {ACM/IEEE} Symposium on
  Logic in Computer Science, {LICS} 2019, Vancouver, BC, Canada, June 24-27,
  2019}}, \bibinfo{publisher}{{IEEE}}, pp. \bibinfo{pages}{1--13},
  \doi{10.1109/LICS.2019.8785713}.

\bibitemdeclare{inproceedings}{DBLP:conf/csfw/FinkbeinerHT19}
\bibitem{DBLP:conf/csfw/FinkbeinerHT19}
\bibinfo{author}{Bernd \surnamestart Finkbeiner\surnameend},
  \bibinfo{author}{Lennart \surnamestart Haas\surnameend} \&
  \bibinfo{author}{Hazem \surnamestart Torfah\surnameend}
  (\bibinfo{year}{2019}): \emph{\bibinfo{title}{Canonical Representations of
  k-Safety Hyperproperties}}.
\newblock In: {\sl \bibinfo{booktitle}{32nd {IEEE} Computer Security
  Foundations Symposium, {CSF} 2019, Hoboken, NJ, USA, June 25-28, 2019}},
  \bibinfo{publisher}{{IEEE}}, pp. \bibinfo{pages}{17--31},
  \doi{10.1109/CSF.2019.00009}.

\bibitemdeclare{inproceedings}{DBLP:conf/stacs/Finkbeiner017}
\bibitem{DBLP:conf/stacs/Finkbeiner017}
\bibinfo{author}{Bernd \surnamestart Finkbeiner\surnameend} \&
  \bibinfo{author}{Martin \surnamestart Zimmermann\surnameend}
  (\bibinfo{year}{2017}): \emph{\bibinfo{title}{The First-Order Logic of
  Hyperproperties}}.
\newblock In \bibinfo{editor}{Heribert \surnamestart Vollmer\surnameend} \&
  \bibinfo{editor}{Brigitte \surnamestart Vall{\'{e}}e\surnameend}, editors:
  {\sl \bibinfo{booktitle}{34th Symposium on Theoretical Aspects of Computer
  Science, {STACS} 2017, March 8-11, 2017, Hannover, Germany}}, {\sl
  \bibinfo{series}{LIPIcs}}~\bibinfo{volume}{66}, \bibinfo{publisher}{Schloss
  Dagstuhl - Leibniz-Zentrum f{\"{u}}r Informatik}, pp.
  \bibinfo{pages}{30:1--30:14}, \doi{10.4230/LIPIcs.STACS.2017.30}.

\bibitemdeclare{article}{DBLP:journals/acta/FridmanP14}
\bibitem{DBLP:journals/acta/FridmanP14}
\bibinfo{author}{Wladimir \surnamestart Fridman\surnameend} \&
  \bibinfo{author}{Bernd \surnamestart Puchala\surnameend}
  (\bibinfo{year}{2014}): \emph{\bibinfo{title}{Distributed Synthesis for
  Regular and Contextfree Specifications}}.
\newblock {\sl \bibinfo{journal}{Acta Informatica}}
  \bibinfo{volume}{51}(\bibinfo{number}{3-4}), pp. \bibinfo{pages}{221--260},
  \doi{10.1007/s00236-014-0194-x}.

\bibitemdeclare{inproceedings}{ggs21}
\bibitem{ggs21}
\bibinfo{author}{Ohad \surnamestart Goudsmid\surnameend}, \bibinfo{author}{Orna
  \surnamestart Grumberg\surnameend} \& \bibinfo{author}{Sarai \surnamestart
  Sheinvald\surnameend} (\bibinfo{year}{2021}):
  \emph{\bibinfo{title}{Compositional Model Checking for Multi-properties}}.
\newblock In \bibinfo{editor}{Fritz \surnamestart Henglein\surnameend},
  \bibinfo{editor}{Sharon \surnamestart Shoham\surnameend} \&
  \bibinfo{editor}{Yakir \surnamestart Vizel\surnameend}, editors: {\sl
  \bibinfo{booktitle}{Verification, Model Checking, and Abstract Interpretation
  - 22nd International Conference, {VMCAI} 2021, Copenhagen, Denmark, January
  17-19, 2021, Proceedings}}, {\sl \bibinfo{series}{Lecture Notes in Computer
  Science}} \bibinfo{volume}{12597}, \bibinfo{publisher}{Springer}, pp.
  \bibinfo{pages}{55--80}, \doi{10.1007/978-3-030-67067-2\_4}.

\bibitemdeclare{book}{DBLP:books/daglib/Hopcroft}
\bibitem{DBLP:books/daglib/Hopcroft}
\bibinfo{author}{John~E. \surnamestart Hopcroft\surnameend},
  \bibinfo{author}{Rajeev \surnamestart Motwani\surnameend} \&
  \bibinfo{author}{Jeffrey~D. \surnamestart Ullman\surnameend}
  (\bibinfo{year}{2001}): \emph{\bibinfo{title}{Introduction to Automata
  Theory, Languages, and Computation, 2nd Edition}}.
\newblock \bibinfo{series}{Addison-Wesley series in computer science},
  \bibinfo{publisher}{Addison-Wesley-Longman}.

\bibitemdeclare{inproceedings}{BN95}
\bibitem{BN95}
\bibinfo{author}{Bakhadyr \surnamestart Khoussainov\surnameend} \&
  \bibinfo{author}{Anil \surnamestart Nerode\surnameend}
  (\bibinfo{year}{1994}): \emph{\bibinfo{title}{Automatic Presentations of
  Structures}}.
\newblock In \bibinfo{editor}{Daniel \surnamestart Leivant\surnameend}, editor:
  {\sl \bibinfo{booktitle}{Logical and Computational Complexity. Selected
  Papers. Logic and Computational Complexity, International Workshop {LCC} '94,
  Indianapolis, Indiana, USA, 13-16 October 1994}}, {\sl
  \bibinfo{series}{Lecture Notes in Computer Science}} \bibinfo{volume}{960},
  \bibinfo{publisher}{Springer}, pp. \bibinfo{pages}{367--392},
  \doi{10.1007/3-540-60178-3\_93}.

\bibitemdeclare{inproceedings}{DBLP:conf/spin/PommelletT18}
\bibitem{DBLP:conf/spin/PommelletT18}
\bibinfo{author}{Adrien \surnamestart Pommellet\surnameend} \&
  \bibinfo{author}{Tayssir \surnamestart Touili\surnameend}
  (\bibinfo{year}{2018}): \emph{\bibinfo{title}{Model-Checking HyperLTL for
  Pushdown Systems}}.
\newblock In \bibinfo{editor}{Mar{\'{\i}}a{-}del{-}Mar \surnamestart
  Gallardo\surnameend} \& \bibinfo{editor}{Pedro \surnamestart
  Merino\surnameend}, editors: {\sl \bibinfo{booktitle}{Model Checking Software
  - 25th International Symposium, {SPIN} 2018, Malaga, Spain, June 20-22, 2018,
  Proceedings}}, {\sl \bibinfo{series}{Lecture Notes in Computer Science}}
  \bibinfo{volume}{10869}, \bibinfo{publisher}{Springer}, pp.
  \bibinfo{pages}{133--152}, \doi{10.1007/978-3-319-94111-0\_8}.

\bibitemdeclare{phdthesis}{DBLP:phd/dnb/Rabe16}
\bibitem{DBLP:phd/dnb/Rabe16}
\bibinfo{author}{Markus~N. \surnamestart Rabe\surnameend}
  (\bibinfo{year}{2016}): \emph{\bibinfo{title}{A Temporal Logic Approach to
  Information-flow Control}}.
\newblock Ph.D. thesis, \bibinfo{school}{Saarland University}.
\newblock
  \urlprefix\url{http://scidok.sulb.uni-saarland.de/volltexte/2016/6387/}.

\bibitemdeclare{inproceedings}{DBLP:conf/banff/Vardi95}
\bibitem{DBLP:conf/banff/Vardi95}
\bibinfo{author}{Moshe~Y. \surnamestart Vardi\surnameend}
  (\bibinfo{year}{1995}): \emph{\bibinfo{title}{An Automata-Theoretic Approach
  to Linear Temporal Logic}}.
\newblock In \bibinfo{editor}{Faron \surnamestart Moller\surnameend} \&
  \bibinfo{editor}{Graham~M. \surnamestart Birtwistle\surnameend}, editors:
  {\sl \bibinfo{booktitle}{Logics for Concurrency - Structure versus Automata
  (8th Banff Higher Order Workshop, Banff, Canada, August 27 - September 3,
  1995, Proceedings)}}, {\sl \bibinfo{series}{Lecture Notes in Computer
  Science}} \bibinfo{volume}{1043}, \bibinfo{publisher}{Springer}, pp.
  \bibinfo{pages}{238--266}, \doi{10.1007/3-540-60915-6\_6}.

\bibitemdeclare{inproceedings}{DBLP:conf/lics/VardiW86}
\bibitem{DBLP:conf/lics/VardiW86}
\bibinfo{author}{Moshe~Y. \surnamestart Vardi\surnameend} \&
  \bibinfo{author}{Pierre \surnamestart Wolper\surnameend}
  (\bibinfo{year}{1986}): \emph{\bibinfo{title}{An Automata-Theoretic Approach
  to Automatic Program Verification (Preliminary Report)}}.
\newblock In: {\sl \bibinfo{booktitle}{Proceedings of the Symposium on Logic in
  Computer Science {(LICS} '86), Cambridge, Massachusetts, USA, June 16-18,
  1986}}, \bibinfo{publisher}{{IEEE} Computer Society}, pp.
  \bibinfo{pages}{332--344}.

\bibitemdeclare{article}{vw94}
\bibitem{vw94}
\bibinfo{author}{Moshe~Y. \surnamestart Vardi\surnameend} \&
  \bibinfo{author}{Pierre \surnamestart Wolper\surnameend}
  (\bibinfo{year}{1994}): \emph{\bibinfo{title}{Reasoning About Infinite
  Computations}}.
\newblock {\sl \bibinfo{journal}{Inf. Comput.}}
  \bibinfo{volume}{115}(\bibinfo{number}{1}), pp. \bibinfo{pages}{1--37},
  \doi{10.1006/inco.1994.1092}.

\bibitemdeclare{article}{wzbp19}
\bibitem{wzbp19}
\bibinfo{author}{Yu~\surnamestart Wang\surnameend}, \bibinfo{author}{Mojtaba
  \surnamestart Zarei\surnameend}, \bibinfo{author}{Borzoo \surnamestart
  Bonakdarpour\surnameend} \& \bibinfo{author}{Miroslav \surnamestart
  Pajic\surnameend} (\bibinfo{year}{2019}): \emph{\bibinfo{title}{Statistical
  Verification of Hyperproperties for Cyber-Physical Systems}}.
\newblock {\sl \bibinfo{journal}{{ACM} Trans. Embed. Comput. Syst.}}
  \bibinfo{volume}{18}(\bibinfo{number}{5s}), pp. \bibinfo{pages}{92:1--92:23},
  \doi{10.1145/3358232}.

\end{thebibliography}

\end{document}